\documentclass[prd,aps,amsfonts,eqsecnum,superscriptaddress,nofootinbib,longbibliography,notitlepage]{revtex4-1}

\usepackage{graphicx}
\usepackage{xcolor}
\usepackage{subcaption}
\usepackage{rotating}
\usepackage{amsmath,amssymb,graphics,amsthm,isomath}
\usepackage{physics }
\usepackage{verbatim}
\usepackage{circuitikz}
\usepackage{tikz}
\usepackage{esint}
\usepackage{blkarray}

\usepackage[colorlinks=true, urlcolor=violet, linkcolor=blue, citecolor=red, hyperindex=true, linktocpage=true]{hyperref}
\usepackage[capitalise,compress]{cleveref}

\newcommand\andy[1]{{[\color{blue} andy: #1]}}

\newcommand{\AO}[1]{{[\color{red}{ AO: #1}}]}
\newtheorem{thm}{Theorem}
\numberwithin{thm}{section}
\newtheorem{cor}[thm]{Corollary}

\newtheorem{obs}[thm]{Observation}
\newtheorem{defn}[thm]{Definition}

\makeatletter
\renewcommand{\p@subsection}{}
\renewcommand{\p@subsubsection}{}
\makeatother

\newcommand{\introem}[0]{Let $G = (\mathcal V, \mathcal C,\mathcal I, \mathcal L, A, B)$ be a circuit. }

\usepackage{mathtools}
\tikzstyle{densely dashed}=          [dash pattern=on 4pt off 3pt]

\usepackage{dsfont}

\begin{document}

\title{Flux-charge symmetric theory of superconducting circuits}
\author{Andrew Osborne}
\email{andrew.osborne-1@colorado.edu}
\affiliation{Department of Physics and Center for Theory of Quantum Matter, University of Colorado, Boulder CO 80309, USA}
\author{Andrew Lucas}
\email{andrew.j.lucas@colorado.edu}
\affiliation{Department of Physics and Center for Theory of Quantum Matter, University of Colorado, Boulder CO 80309, USA}

\begin{abstract}
The quantum mechanics of superconducting circuits is derived by starting from a classical Hamiltonian dynamical system describing a dissipationless circuit, usually made of capacitive and inductive elements. However, standard approaches to circuit quantization treat fluxes and charges, which end up as the canonically conjugate degrees of freedom on phase space, asymmetrically. 
By combining intuition from topological graph theory with a recent symplectic geometry approach to circuit quantization, we present a theory of circuit quantization that treats charges and fluxes on a manifestly symmetric footing. For planar circuits, known circuit dualities are a natural canonical transformation on the classical phase space. We discuss the extent to which such circuit dualities generalize to non-planar circuits.

\end{abstract}

\maketitle

\tableofcontents

\section{Introduction}

An ordinary non--dissipative circuit, when sufficiently cooled such that all wires become superconducting, will exhibit macroscopic manifestations of quantum mechanics with a finite number of degrees of freedom.
Such circuits are referred to as superconducting circuits; they are of interest for a variety of reasons, not least of which is the design and implementation of qubits for use in a future  quantum computer \cite{krantz_quantum_2019,kjaergaard_superconducting_2020,devoret_fqi,brayvi_the_future,kitaev2006protected}.
To that end, the theory of superconducting circuits has been very successful in producing a number of qubits with desirable properties, including the transmon, $0$--$\pi$ and fluxonium qubits \cite{koch_2017,gyenis_2021,manucharyan_fluxonium_2009,kitaev_encoding}. 
On more fundamental grounds, superconducting circuits are also of use for quantum simulation of exotic physics, such as particle motion on spaces with negative curvature \cite{Kollar2019}. 

In the theory of superconducting circuits, charges and fluxes take the role of momenta and positions in classical Hamiltonian mechanics \cite{burkard_multilevel_2004}.  More formally, they are canonically conjugate degrees of freedom. 
For simple circuits, it is often straightforward to identify the conjugate pairs of charges and fluxes, and then promote these variables to operators, as in textbook quantum mechanics.  This procedure is referred to as circuit quantization. 
However, for complicated circuits, namely those containing both phase slips \cite{mooij_superconducting_2006} and Josephson junctions \cite{Astafiev2012,belkinqps}, the standard approach \cite{Ciani:2023ubt} to circuit quantization relies on using a \emph{Lagrangian formulation} of the problem that breaks the explicit structure, such as the presence of canonical transformations that do not change the underlying structure behind Hamiltonian mechanics.  In fact, only very recently have some authors \cite{osborne2023symplectic,Parra-Rodriguez:2023ykw,ParraRodriguez2022canonical} begun to develop an intrensically Hamiltonian \cite{goldstein1980classical} approach to describing the classical (and quantum) mechanics of superconducting circuits.   Still, these approaches have continued to treat charge and flux on an arguably asymmetric footing.


This paper presents a theory of circuit quantization that is manifestly symmetric under the exchange of charge and flux degrees of freedom, and treats capacitors and inductors on an equal footing.  An advantage of this approach is that it makes certain ``mysterious" properties of the algorithmic method of \cite{osborne2023symplectic}, such as the appearance of conserved quantities, more manifest. Moreover, our approach provides an extremely transparent description of \emph{circuit dualities} \cite{ulrich_dual_2016,Kerman2013} for planar circuits, and -- in some special cases -- for non-planar circuits as well.   

\section{Classical theory of circuits}
  We consider circuits made out of elements that are either purely inductive or purely capacitive.  An inductive element has a constitutive relation of the form 
\begin{equation}
  I = \frac{\partial E}{\partial \phi},
\end{equation}
where $E$ is the energy stored in the element, and $\phi$ is the magnetic flux through it.  Likewise, a circuit element with a constitutive relation 
\begin{equation}
  V = \frac{\partial E}{\partial q} 
\end{equation}
is a capacitor: here $q$ is the charge accumulated across the capacitor, while $V$ is the voltage.  Note that these definitions include arbitrary nonlinear elements such as the Josephson junction and the quantum phase slip junction. 
From the perspective of circuit quantization, these nonlinear elements are no obstruction to any of the results that we present. 

Loosely speaking, $I = \dot q$ and $V = \dot \phi$ (here dots denote time derivatives).  This suggests that inductors and capacitors host conjugate degrees of freedom in a superconducting circuit.


\subsection{Symplectic geometry approach to circuit quantization}

It is our goal to produce a quantization perscription for superconducting circuits that treats capacitors and inductors (equivalently charges and fluxes) on equal footing.
To ensure full generality in our construction, we will begin with a short review of a  recently developed \cite{osborne2023symplectic} symplectic approach to quantization of circuits, which is applicable to arbitrarily nonlinear circuits built out of inductors and capacitors.  
Besides having the advantage of being more general (the new approach can quantize circuits for which a standard Lagrangian does not exist), we will show in this paper that the approach of \cite{osborne2023symplectic} also beautifully encodes the mathematics of circuit duality, when re-written in a flux-charge-symmetric manner.

In \cite{osborne2023symplectic}, it was found that a  Hamiltonian dynamical system -- namely, a Hamiltonian function $H$ and a symplectic form and/or Poisson bracket on a suitable phase space --  is neatly encoded in a simple Lagrangian that depends on the incidence matrix of a circuit.  A circuit can be thought of mathematically as a directed graph with vertices $v$ and directed edges $e$, which have a start and end vertex.  On each edge $e$, we place either an inductive or capacitive element whose energy $E$ depends either on the flux difference $\phi_e$ across the circuit element, or the charge $q_e$ which has accumulated due to current flow across the edge.  The directedness of each edge is important to capture the correct sign of $\phi_e$ and $q_e$.  
For the remainder of this manuscript, we will use the symbol $\mathcal E $ to refer to the set of edges in a circuit. Generally, the circuit to which $\mathcal E$ pertains will be clear from context.
We define the incidence matrix
\begin{equation}
    A_{ev} = \begin{cases}
       1 & e \text{ arrives at } v \\ 
       -1 & e \text{ leaves from } v \\
       0 & \text{otherwise}
    \end{cases}
\end{equation}
which encodes a convention for positive current and voltage across an edge. 
There is another closely related matrix, $\Omega$, called the reduced incidence matrix, which is a restriction of $A$ to edges containing only capacitors.  
In terms of $\Omega$, the action encoding the dynamics of $G$ may be written 
\begin{equation}\label{eqn:action}
    S = \int \mathrm d t \left[ 
    \sum_{e \in \mathcal C, v} q_e \Omega_{ev} \dot\phi_v - \sum_{e \in \mathcal C} E^C_{e}(q_e) - \sum_{e \in \mathcal I} E^L_e (A_{ev}\phi^v)
    \right]
\end{equation}
with capacitive branches in the set $\mathcal C$ and inductive branches in the set $\mathcal I$.   Note that $\mathcal{E}=\mathcal{C}\cup\mathcal{I}$.
Here, $E^C_e$ ($E^L_e$) is the energy of the capacitor (inductor) on branch $e$ for a given configuration. 

Quantization of a circuit is achieved by enumerating and removing the null vectors of $\Omega$ and the corresponding variables.  The resulting invertible matrix $\Omega$ can be related to a symplectic form and Poisson bracket, which then formally lead to a canonical quantization prescription of the circuit \cite{osborne2023symplectic}.

\subsection{Topology of a graph}
In this paper, we will need to review a little more graph theory than was presented above.  Indeed, in the language of mathematics, circuits are directed graphs where the edge set $\mathcal{E}=\mathcal{C}\cup\mathcal{I}$ has a decomposition into capacitors and inductors.  
We now briefly review some of the relevant graph theory, and refer the reader to Appendix \ref{app:graphthry} for a more formal discussion.

The most important aspect of a graph for us here will be the structure of cycles, or loops. For an undirected graph, it is possible to define 
cycle addition \cite{gross_topological,hatcher_topology} as an additive operation on a vector space over the field $\mathbb{Z}_2$. This operation imbues us with a natural notion of linear independence between cycles.  This notion of independence will remain relevant for circuits, which are directed graphs, since the ``arrow" on each edge is used to keep track of the direction of current flow, but not whether an object is a loop or not.

For a circuit with $|\mathcal E|$ edges, and $|\mathcal V|$ vertices, there are $|\mathcal E| - |\mathcal V| + 1$ linearly independent loops, assuming that the circuit is connected.  For reasons that will become clear, we will find it expedient to find a basis for this space with \emph{one redundant loop}: i.e., we choose $|\mathcal E| - |\mathcal V| + 2$ loops and collect them into set \begin{equation}
    \mathcal L = \{l_1, l_2, \dots, l_{|\mathcal E| - |\mathcal V| + 2}\}.
\end{equation}   
To each loop assign an orientation, such that an edge in the loop $l$ may either be oriented with $l$ or against $l$.  Such assignments are encoded in an \textbf{orientation matrix}
\begin{equation}
    B_{le} = 
    \begin{cases}
        1 & l \text{ and } e \text{ are oriented alike }\\
        -1 & l \text{ and } e \text{ are oriented unalike} \\ 
        0 & e \text{ does not lie on the boundary of } l
    \end{cases}.
\end{equation}
For an example, consider the graph drawn in Fig.~\ref{fig:orientation}. The orientation matrix for this graph is given by
\begin{equation}
    B = 
    \begin{blockarray}{ccccccccccc}
     e_1 & e_2 & e_3 & e_4 & e_5 & e_6 & e_7 & e_8 & e_9 & e_{10} & \\ 
    \begin{block}{(cccccccccc)@{\hspace{7pt}}c}
    1 & 0 & 1 & 0 & 0 & 0 &-1 & 0 & 0 & 0  & l_1 \\
    0 &-1 & 0 & 0 & 1 & 0 & 1 & 1 &-1 &-1 & l_2 \\
    0 & 0 &-1 & 1 & 0 & 1 & 0 &-1 & 0 & 0 & l_3 \\
    -1& 1 & 0 &-1 &-1 &-1 & 0 & 0 & 1 & 1 & l_4  \\
    \end{block}
    \end{blockarray}
\end{equation}
Since we have defined $\mathcal L$ in such a way that it always contains $|\mathcal E| - |\mathcal V| + 2$ elements, it is always the case that 
\begin{equation}\label{eqn:eulermod}
    |\mathcal L| - |\mathcal E| + |\mathcal V| = 2.
\end{equation}
Euler's formula should be called to mind by (\ref{eqn:eulermod}). 
As we will see, this is no accident. 
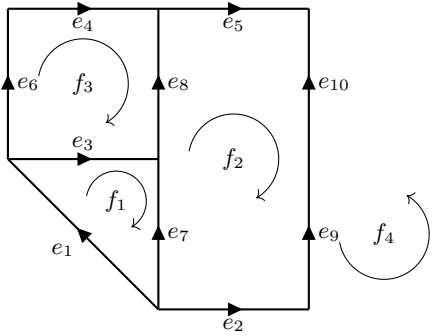
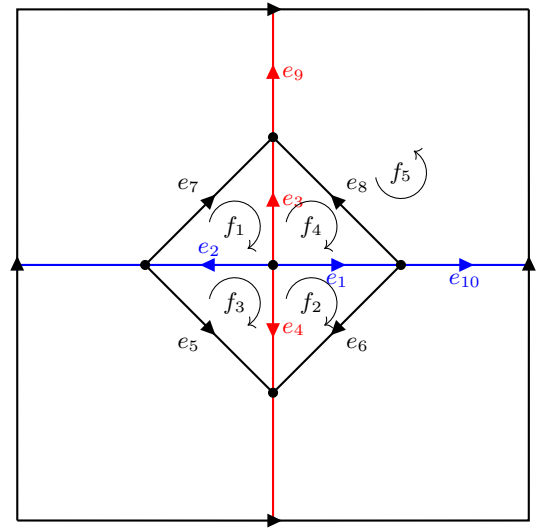
\begin{figure}
\begin{subfigure}{0.5\linewidth}
\begin{circuitikz}[scale=2]
    \draw[thick] (1,0) to[short,i_=$e_2$] (2,0);
    \draw[thick] (0,1) to[short,i^=$e_3$] (1,1);
    \draw[thick] (0,2) to[short,i_=$e_4$] (1,2) to[short,i_=$e_5$] (2,2);
    \draw[thick] (1,0) to[short,i^=$e_1$] (0,1);
    \draw[thick] (0,1) to[short,i_=$e_{6}$] (0,2);
    \draw[thick] (1,0) to[short,i_=$e_{7}$] (1,1) to[short,i_=$e_{8}$] (1,2);
    \draw[thick] (2,0) to[short,i_=$e_{9}$] (2,1) to[short,i_=$e_{10}$] (2,2);
    \draw[thin, <-] (0.72,0.72)node{$f_1$}  ++(-60:0.2) arc (-60:170:0.2);
    \draw[thin, <-] (1.5,1)node{$f_2$}  ++(-60:0.3) arc (-60:170:0.3);
    \draw[thin, <-] (2.5,0.5)node{$f_4$}  ++(60:0.3) arc (60:-170:0.3);
    \draw[thin, <-] (0.5,1.5)node{$f_3$}  ++(-60:0.3) arc (-60:170:0.3);
\end{circuitikz}
\caption{An embedding of a graph with orientations chosen for faces and branches. The outer face is labeled $f_4$. All faces are oriented alike if we consider this plane to be a patch of a sphere (namely we identify points at infinity as the same point).
}
\end{subfigure}
\begin{subfigure}{.49\linewidth}
\begin{circuitikz}[scale=1.7]
        \draw[thick,blue] (0,0) to[short,i_=$e_1$,color=blue] (1,0);
        \draw[thick,blue] (0,0) to[short,i_=$e_2$,color=blue] (-1,0);
        \draw[thick,red] (0,0) to[short,i_=$e_3$,color=red] (0,1);
        \draw[thick,red] (0,0) to[short,i^=$e_4$,color=red] (0,-1);
        \draw[thick] (-1,0) to[short,i_=$e_5$] (0,-1);
        \draw[thick] (1,0) to[short,i^=$e_6$] (0,-1);
        \draw[thick] (-1,0) to[short,i^=$e_7$] (0,1);
        \draw[thick] (1,0) to[short,i_=$e_8$] (0,1);
        
        \draw[thick,blue] (1,0) to[short,i_=$e_{10}$,color=blue] (2,0);
        \draw[thick,blue] (-1,0) to[short,color=blue] (-2,0);
        \draw[thick,red] (0,1) to[short,i_=$e_9$,color=red] (0,2);
        \draw[thick,red] (0,-1) to[short,color=red] (0,-2);

        \draw[thick] (-2,-2) to[short,i_=$ $] (-2,2) to[short,i_=$ $ ] (2,2); 
        \draw[thick] (-2,-2) to[short,i_=$ $] (2,-2) to[short,i_=$ $] (2,2); 

        \filldraw[black] (0,0) circle (1pt); 
        \filldraw[black] (1,0) circle (1pt); 
        \filldraw[black] (0,1) circle (1pt); 
        \filldraw[black] (0,-1) circle (1pt); 
        \filldraw[black] (-1,0) circle (1pt); 
    \draw[thin, <-] (-0.3,0.3)node{$f_1$}  ++(-60:0.2) arc (-60:170:0.2);
    \draw[thin, <-] (0.3,-0.3)node{$f_2$}  ++(-60:0.2) arc (-60:170:0.2);
    \draw[thin, <-] (-0.3,-0.3)node{$f_3$}  ++(-60:0.2) arc (-60:170:0.2);
    \draw[thin, <-] (0.3,0.3)node{$f_4$}  ++(-60:0.2) arc (-60:170:0.2);
    \draw[thin, <-] (1,.72)node{$f_5$}  ++(60:0.2) arc (60:-170:0.2);
\end{circuitikz}
\caption{An embedding of $K_5$ on the torus. The vertical border at the top of the drawing is to be identified with the vertical border at the bottom of the drawing. Likewise, the horizontal boundaries are to be identified. A valid choice of ``topological loops" is given by the blue and red edges. }
\end{subfigure}
\caption{Examples of faces and edges in a (a) planar and (b) non-planar graph.}\label{fig:orientation}
\end{figure}

Every graph (and thus every circuit) can be drawn (more formally, embedded) on some closed, orientable surface\footnote{The surfaces one should keep in mind are the Riemann surface of genus $g$, including the sphere ($g=0$) and the torus ($g=1$). Such surfaces will always be sufficient since the topology of two dimensional orientable manifolds is fully characterized by their genus.} with no crossing edges. 
One can always find such a surface, wherein this drawing of a graph partitions the surface of embedding into some number of regions isomorphic to the unit disk. 
We will call every such two dimensional disk a \emph{face} of the circuit.
For a particular drawing of a circuit on a surface of genus $g$, we will write that the number of faces induced by the drawing is $|\mathcal F|$. 
Euler's formula reads 
\begin{equation}\label{eqn:eulergood}
    |\mathcal F| - |\mathcal E| + |\mathcal V| = 2 - 2g. 
\end{equation}
Evidently, (\ref{eqn:eulergood}) together with (\ref{eqn:eulermod}) implies that 
\begin{equation}
    |\mathcal L| = |\mathcal F| + 2 g.
\end{equation}
Of course, one may have suspected from the outset that $|\mathcal L|$ was closely related to $|\mathcal F|$ since the boundary of a unit disk is a circle, a one--dimensional object. 

The relationship between faces and loops is an exact correspondence (bijection) in the case of planar graphs. 
More precisely, the set of loops at the boundary of some face in a drawing of a planar graph on the sphere (or the plane) number $|\mathcal E| - |\mathcal V| + 2$, and every edge exists in the common boundary of exactly two faces. 
Moreover, it is possible to demand that all the loops chosen in this way are ``oriented alike" such that planar graphs satisfy\footnote{We emphasize that this particularly simple redundancy in $|\mathcal L|$ is a consequence of choosing loops in correspondence to faces. In principle, it is possible to choose a spanning set of loops with a less trivial redundancy, but this complicates calculations and provides no simplification. Any choice of loops for a planar graph can be shown to correspond to the faces of \emph{some} drawing, so this choice is also perfectly general. }
\begin{equation}\label{eqn:faceground}
    \sum_{l \in \mathcal L} B_{le} = 0.
\end{equation}
As we will see, choosing loops in correspondence with faces will be of great utility for us, even for nonplanar graphs. 
It will always be our convention to choose $|\mathcal F|$ loops corresponding to the faces of some embedding of a circuit, and we will also always choose them to be oriented alike. 

Another important matter of convention arises when considering the sign of $B_{le}$ when it is nonzero. Consider the figure in Fig.~\ref{fig:orientation}(a), and the loop consisting of edges $e_1$, $e_3$, and $e_7$. 
Intuitively, it should be the case that a suitable definition of $B$ for this graph has the property that $B_{l_1 e}$ is nonzero only for $e$ in $\{e_1,e_3,e_7\}$. 
With this preference in mind, it remains to fix the sign of $B_{l_1 e}$ for such edges $e$. Our convention will be that an orientation should be chosen for $l_1$ so that $B_{l_1 e}$ is equal to one if $e$ is oriented like $l_1$ and negative one otherwise. In this particular example (as determined by the semicircular arrow surrounding the label $f_1$), the matrix elements of $B$ should be $B_{l_1e_1} = B_{l_1 e_3} = - B_{l_1e_7} = 1$.

For nonplanar graphs, it is always possible to choose an embedding such that every edge appears on the boundary of precisely two faces. However, it is not always possible to choose an embedding such that every edge appears in precisely two loops in $\mathcal L$. 
As an example, consider edge $e_9$ in the circuit drawn in Fig.~\ref{fig:orientation}b. $e_9$ appears on the boundary of $f_5$, and on the boundary of no other face. However, $e_9$ appears on the boundary of $f_5$ twice.  When this occurs, we take the convention that $B_{f_5e_9}=0$, since ``each side" of $e_9$ is seen by the same face, and $+1-1=0$.

In order to define an orientation matrix for a nonplanar graph, one needs to choose apropriate topological cycles and repeat the procedure. For the graph drawn in Fig.~\ref{fig:orientation}(b), one topological loop is drawn in blue, while the other is drawn in red. 
The orientation matrix for this graph is given by 
\begin{equation}\label{eqn:bmat}
    B =
    \begin{blockarray}{ccccccccccc}
    e_1 & e_2 & e_3 & e_4 & e_5 & e_6 & e_7 & e_8 & e_9 & e_{10} & \\ 
    \begin{block}{(cccccccccc)@{\hspace{7pt}}c}
         0 & 1 & -1 & 0 & 0 & 0 & 1 & 0 & 0 & 0 & l_1 \\ 
        1 & 0 & 0 & -1 & 0 & 1 & 0 & 0 & 0 & 0 & l_2 \\ 
        0 & -1 & 0 & 1 & -1 & 0 & 0 & 0 & 0 & 0& l_3 \\ 
        -1 & 0 & 1 & 0 & 0 & 0 & 0 & -1 & 0 & 0& l_4 \\ 
        0 & 0 & 0 & 0 & 1 & -1 & -1 & 1 & 0 & 0& l_5 \\ 
        1 & -1 & 0 & 0 & 0 & 0 & 0 & 0 & 0 & 1 & l_6 \\ 
        0 & 0 & 1 & -1 & 0 & 0 & 0 & 0 & 1 & 0 & l_7 \\ 
    \end{block}
    \end{blockarray}
\end{equation}
It turns out that $B$ (for any drawing of $G$) has the property that 
\begin{equation}\label{eqn:shortexactmain}
    \sum_{e} B_{le}A_{ev} = 0.
\end{equation}
In the case of (circuits that can be embedded as) planar graphs, this is an elementary result from algebraic topology \cite{hatcher_topology}: $B$ and $A$ are the discrete differential acting on 0-forms (defined on vertices) and 1-forms (defined on edges) respectively (see Appendix \ref{app:graphthry}).  For non-planar graphs, the matrix $B$ includes the discrete differential acting on 1-forms, \emph{and} the non-trivial elements of the first cohomology group.  While we aren't aware of any elegant geometric interpretation of this combined object, (\ref{eqn:shortexactmain}) continues to hold.
Note, however that for nonplanar graphs, the condition (\ref{eqn:faceground}) cannot hold. However, there is some analogous expression corresponding to the set of ``face--like" loops in $\mathcal L$. In the case of the circuit drawn in Fig.~\ref{fig:orientation}, (\ref{eqn:bmat}) admits 
\begin{equation}
    \sum_{i = 1}^5 B_{l_i e} = 0.
\end{equation}
It is always possible to incorporate some such condition into a choice of loops in $\mathcal L$.

\begin{figure}
  \begin{circuitikz}[scale=3]
    \draw[thick] (0,0) to[inductor,i^=$e_1$] (1,2) to[inductor,i^=$e_2$] (2,0) to [capacitor,i^=$e_3$] (0,0);
    \draw[thick] (0,0) to[capacitor,i_=$e_4$] (1,.75) to[capacitor,i_=$e_5$] (2,0);
    \draw[thick] (1,2) to[inductor,i_=$e_6$] (1,.75);
    \filldraw[black] (0,0) circle (0.5pt) node[anchor=north east]{$v_1$};
    \filldraw[black] (1,2) circle (0.5pt) node[anchor=south]{$v_2$};
    \filldraw[black] (2,0) circle (0.5pt) node[anchor=north west]{$v_3$};
    \filldraw[black] (1,.75) circle (0.5pt) node[anchor=north]{$v_4$};
    \draw[thin, <-] (0.7,0.88)node{$f_1$}  ++(-60:0.15) arc (-60:170:0.15);
    \draw[thin, <-] (1.3,0.88)node{$f_2$}  ++(-60:0.15) arc (-60:170:0.15);
    \draw[thin, <-] (1,0.35)node{$f_3$}  ++(-60:0.15) arc (-60:170:0.15);
    \draw[thin, <-] (0,1)node{$f_4$}  ++(60:0.3) arc (60:-170:0.3);
  \end{circuitikz}
  \caption{A  circuit with four faces (counting the external face), six edges, and four vertices. For this circuit, $\mathcal  I = \{e_1,e_6,e_2\}$ and $\mathcal C = \{e_3,e_4,e_5\}$. To see that all faces are ``oriented alike", imagine the circuit embedded on the sphere: all faces are oriented into the sphere.}\label{fig:excir}
\end{figure}
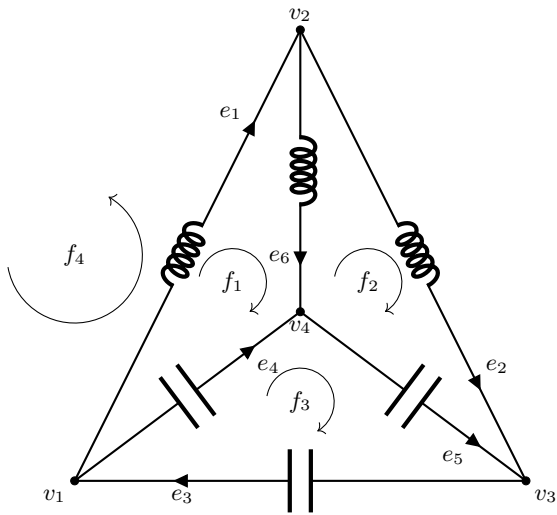

We now introduce our third and final matrix. First, note that it follows from (\ref{eqn:shortexactmain}) that, for a set $P \subset \mathcal E$, 
\begin{equation}\label{eqn:splitter}
   \sum_{e \in P} B_{le}A_{ev} = - \sum_{e \not\in  P} B_{le}A_{ev}.
\end{equation}
This fact will be of great use to us.  In particular, there will be a natural partition $P$ for a circuit: simply count only the capacitive, or only the inductive, edges!  Hence for any circuit, we define the \textbf{connection matrix}
\begin{equation}
  M_{lv} =  \frac{1}{2}\sum_{e \in \mathcal C} B_{le} A_{ev} - \frac{1}{2}\sum_{e \in \mathcal I} B_{le}A_{ev}.
\end{equation}

Consider the circuit drawn in Fig.~\ref{fig:excir}, which we will call $G$.
In the following discussion, we will repeatedly return to $G$ as an example for the sake of concreteness. 
As we will see, the connection matrix of a circuit has a number of remarkable properties. 
For the circuit $G$ in Fig.~\ref{fig:excir},
 \begin{equation}
   A = 
   \begin{blockarray}{ccccc}
    v_1 & v_2 & v_3 & v_4  & \\ 
       \begin{block}{(cccc)@{\hspace{7pt}}c}
 -1 & 1 & 0 & 0 & e_1 \\
 0 & -1 & 1 & 0 & e_2 \\
 1 & 0 & -1 & 0& e_3 \\
  -1 & 0 & 0 & 1& e_4 \\
  0 & 0 & 1 & -1& e_5 \\
     0 & -1 & 0 & 1& e_6 \\
       \end{block}
   \end{blockarray}
\end{equation}
\begin{equation}
  B = 
  \begin{blockarray}{ccccccc}
   e_1 & e_2 & e_3 & e_4 & e_5 & e_6 &\\ 
  \begin{block}{(cccccc)@{\hspace{7pt}}c}
 1 & 0 & 0 & -1 & 0 & 1& l_1 \\
 0 & 1 & 0 & 0 & -1 & -1 & l_2 \\
 0 & 0 & 1 & 1 & 1 & 0 & l_3 \\
 -1 & -1 & -1 & 0 & 0 & 0 & l_4\\
  \end{block}
  \end{blockarray}
\end{equation}
and 
\begin{equation}
  M = 
  \begin{blockarray}{ccccc}
   v_1 & v_2 & v_3 & v_4 &\\ 
  \begin{block}{(cccc)@{\hspace{7pt}}c}
 1 & 0 & 0 & -1& l_1 \\
 0 & 0 & -1 & 1& l_2 \\
 0 & 0 & 0 & 0& l_3 \\
 -1 & 0 & 1 & 0& l_4 \\
 \end{block}
  \end{blockarray}
\end{equation}
It is possible to build a robust theory of circuit quantization using $B$ in place of the incidence matrix of a circuit $A$ or, indeed, in conjunction with $B$.  
While this point is conceptually appealing on its own, it also turns out to be necessary to understand circuit duality in full generality.
The main result of this paper is that the matrix $M$, containing information inherited from both $A$ and $B$, encodes the symplectic form of the classical phase space, which may be used to quantize a Hamiltonian 
without requiring the user to proclaim that charges are more fundamental than fluxes or vice--versa. 
\subsection{A symmetric circuit Lagrangian}

The charge accumulated on a particular capacitive branch is given by $q_e = \sum_{l} q_l B_{le}$, which may be used to connect to the formalism of \cite{osborne2023symplectic}, as we will see later. Moreover,  these loop charge variables are precisely the ones described in \cite{ulrich_dual_2016}, with the exception of the fact that we ``count" the external loop. Of course, this is only a matter of convenience since one is always free to fix one loop charge to vanish as a gauge degree of freedom, in direct analogy with one's freedom to choose a ground. 
Now, for any circuit, 
\begin{equation}\label{eqn:splitvars}
  \frac{1}{2}\sum_{e \in \mathcal C} q_l B_{le} A_{ev} \dot\phi_v -  \frac{1}{2}\sum_{e \in \mathcal I} q_l B_{le}A_{ev} \dot\phi_v = 
  \sum_{e \in \mathcal C} q_l B_{le}A_{ev}\dot\phi_v =  -\sum_{e \in \mathcal I} q_l B_{le}A_{ev} \dot\phi_v, 
\end{equation}
according to (\ref{eqn:splitter}).
While (\ref{eqn:splitvars}) contains only trivial mathematical information, we notice that the  first expression in (\ref{eqn:splitvars}) is not preferential toward inductive or capacitive branches, while the second two expressions \emph{are}.

For the remainder of this manuscript, we will write $E_e^L$ ($E_e^C$) to denote the energy of an inductive (capacitive) branch  $e$  in some circuit. We assume that such energies may only depend on the flux across (the charge on) some branch.\footnote{There do exist ideal non-reciprocal circuit elements such as gyrators that are not of this form \cite{}. We leave any possible generalization of our results to such circuits to future work.} Explicitly, we require that, for $e = (v_1,v_2)$, $E_e^I$ depends only on  $\phi_{v_1} - \phi_{v_2}$. Analogously, for a branch adjacent to the loops $l_1$ and $l_2$, we require that $E_e^C$ may depend only  on $q_{l_1} - q_{l_2}$. 
As we will show shortly, this demand is equivalent to the imposition of Kirchoff's rules. 
We claim that the Lagrangian 
\begin{equation}\label{eqn:motherlag}
  L = \sum_{l,v} q_l M_{lv} \dot\phi_v - \sum_{e \in \mathcal C} E_e^{C}\left(\sum_{l}q_lB_{le}\right) - \sum_{e \in  \mathcal I}  E_e^{L}\left(\sum_{v}A_{ev}\phi_v\right)
\end{equation}
is a general Lagrangian which can be used to describe arbitrary LC circuits, and which contains a Hamiltonian which is easily quantizable.

Let us examine the equations of motion implied by the principle of least action applied to $S = \int \mathrm dt 
\; L$: 
\begin{equation}\label{eqn:eoms}
    \begin{aligned}
        0 &= \frac{\delta S}{\delta q_l} = \sum_v M_{lv} \dot\phi_v - \sum_{e \in \mathcal C}\frac{\partial E_e^C}{\partial q_l}\\
        0 &= \frac{\delta S}{\delta \phi_v} = - \sum_l\dot q_l M_{lv}   - \sum_{e \in \mathcal I} \frac{\partial E_e^L}{\partial \phi_v}.
    \end{aligned}
\end{equation}
We rewrite the first line using the definition of $M$, to see that 
\begin{equation}\label{eqn:voltagerule}
   0 = -\sum_{e \in \mathcal I,v} B_{le}A_{ev} \dot\phi_v  - \sum_{e \in \mathcal C} \frac{\partial E^C_e}{\partial q_l}.
\end{equation}
Another way of writing (\ref{eqn:voltagerule}) is 
\begin{equation}\label{eqn:voltagerule2}
    0 = \sum_{e \in \mathcal C, v} B_{le}A_{ev} \dot\phi_v - \sum_{e \in \mathcal C} \frac{\partial E_e^C}{\partial q_l}.
\end{equation}
Further, since we have demanded that $E_e^C$ may only depend on $\sum_{l}B_{le}q_l$, we can see that (\ref{eqn:voltagerule2}) implies 
\begin{equation}\label{eqn:capcon}
    0 = \sum_{e \in \mathcal C}B_{le} \left[\sum_v A_{ev}\dot\phi_v - \frac{\partial E_e^C(q)}{\partial q}\right].
\end{equation}
Given that $\sum_v A_{ev}\dot\phi_v$ can be interpreted as a voltage difference across branch $e$, we see that (\ref{eqn:capcon}) is compatible with the constitutive relation for capacitors: \begin{equation}\label{eq:capconst}
    \sum_v A_{ev} \dot\phi_v = \frac{\partial E_e^C(q)}{\partial q}.
\end{equation} In fact, we will show in Corollary \ref{cor:bnullv} that (up to a certain exception to be discussed later), the \emph{only} solution to (\ref{eqn:capcon}) is (\ref{eq:capconst}).
Substituting 
\begin{equation}
    \sum_{e \in \mathcal C} \frac{\partial E_e^C}{\partial q_l} = \sum_{e \in \mathcal C,v}  B_{le}A_{ev}\dot\phi_v, 
\end{equation}
we can see that (\ref{eqn:voltagerule}) becomes 
\begin{equation}
    \sum_{e,v} B_{le}A_{ev}\dot\phi_v = 0.  
\end{equation}
In other words, (\ref{eqn:voltagerule}) is also consistent with Kirchoff's voltage rule around a particular loop $l$.

Likewise, we rewrite the second line of (\ref{eqn:eoms}) as 
\begin{equation}\label{eqn:currentrule}
  0 = -\sum_{l, e \in \mathcal C} \dot q_l B_{le}A_{ev} - \sum_{e \in \mathcal I} \frac{\partial E^L_e}{\partial  \phi_v} = \sum_{l, e\in \mathcal I} \dot q_l B_{le} A_{ev}  - \sum_{e \in \mathcal I} \frac{\partial E_e^L}{\partial \phi_v}.
\end{equation}
A similar manipulation to the one in (\ref{eqn:voltagerule}) and (\ref{eqn:voltagerule2}) can be used to show that (\ref{eqn:currentrule}) correctly incorporates the constitutive relation for inductors, as well as Kirchoff's current law on a particular node $v$.


\subsection{A Hamiltonian theory from a Lagrangian theory}

Classically, the dynamics of a circuit are entirely determined by the Lagrangian in (\ref{eqn:motherlag}). In order to quantize, however, we must produce a symplectic form, which will in turn imply quantum commutation relations between charge and flux operators. 
Practically, this is a matter of enumerating the null vectors of $M$ and integrating out related non--dynamical variables as was done in \cite{osborne2023symplectic}.
One of the notable properties of the formalism at hand is that the null vectors of $M$ do not depend asymmetrically on inductors and capacitors. 

In lieu of a formal discussion, we will focus on a particular example to demonstrate the simplicity of our results. For a formal discussion, see Appendix \ref{app:symmquant}.
For the time being, we will consider the circuit drawn in Fig. \ref{fig:excir}, which we will call $G$. We will suppose that all circuit elements are linear and equal, though we emphasize that this is purely a matter of convenience. 
The Lagrangian for $G$ is 
\begin{equation}\label{eqn:glag}
  \begin{aligned}
    L &= (q_{l_3}-q_{l_4})(\dot\phi_{v_1} - \dot \phi_{v_3}) + (q_{l_3}-q_{l_1})(\dot\phi_{v_{4}} - \dot\phi_{v_1}) + (q_{l_3} - q_{l_2})(\dot\phi_{v_3} - \dot\phi_{v_4})  \\ 
      &- \frac{1}{2 C} (q_{l_3} - q_{l_4})^2 - \frac{1}{2 C}(q_{l_3} - q_{l_1})^2 - \frac{1}{2 C}(q_{l_3} - q_{l_2})^2 \\ 
  &- \frac{1}{2 L}(\phi_{v_2} - \phi_{v_1})^2 - \frac{1}{2 L}(\phi_{v_3} - \phi_{v_2}) - \frac{1}{2 L} (\phi_{v_4} - \phi_{v_2})^2.
  \end{aligned}
\end{equation}

Left null vectors of $M$ correspond to cycles either made entirely of capacitors or entirely of inductors. In $G$, there is only a single such cycle, namely the cycle bounding the loop $l_3$. Constraints of this kind may involve the resolution of some constraint as is the case now:
\begin{equation}\label{eqn:kcons}
  0 = \frac{\delta S}{\delta q_{l_3}} = \frac{1}{ C} (q_{l_3} - q_{l_4}) + \frac{1}{ C}(q_{l_3} - q_{l_1}) + \frac{1}{ C}(q_{l_3} - q_{l_2}).
\end{equation}
The constraint in (\ref{eqn:kcons}) corresponds to the demand that voltage must vanish about a loop, and further that $q_{l_3}$ must be determined by the other loop charges. In other words, $q_{l_3}$ is not dynamical. 
If, instead of capacitors, there were a cycle of inductors, so that there was no voltage constraint, the relevant $q$ variable would simply be cancelled in $L$. This corresponds to a row of $M$ which is identically zero in every entry. 

On the other hand, right null vectors of $M$ correspond to \emph{cuts} of $G$ which are entirely inductive or entirely capacitive. 
A cut of $G$ is a set of edges that connect a subset of $\mathcal V$ to its complement. 
In $G$, there is a single such example, namely the cut consisting of the branches incident upon  $v_2$. 
Notice that $\dot\phi_{v_2}$ does not appear in (\ref{eqn:glag}) at all, so $\phi_{v_2}$ plays the role of a Lagrange multiplier. 
Recognizing that 
\begin{equation}\label{eqn:icons}
  0 = \frac{\delta S}{\delta \phi_{v_2}}
\end{equation}
gives rise to another constraint. This constraint may be interpreted as a demand that current must be conserved at the node $v_2$, and that 
$\phi_{v_2}$ must therefore be fixed in terms of other fluxes. 

There is another pair of null vectors which are omnipresent in the theory of circuits. Namely, energies may only depend upon differences of loop current (node flux), so the sums $\sum_v \phi_v$ and $\sum_{l \in \mathcal F} q_l$ are fixed by a so--called ``gauge freedom"\footnote{We remark that $\sum_{l \in \mathcal L}q_l $ is not necessarily nondynamical.  This leads to subtleties when discussing circuit dualities, which we will discuss in Section \ref{sec:duality}.}. 
Explicitly,
\begin{equation}
  0 = \sum_v \frac{\delta S}{\delta \phi_v} = \sum_l \frac{\delta S}{\delta q_l}.
\end{equation}
Since both of these functional derivatives vanish identically, a constraint can never be imposed by these particular null vectors. 

After resolving the constraints in (\ref{eqn:kcons}) and (\ref{eqn:icons}), we are left with 
\begin{equation}
  L = Q_1 \dot\Phi_1 + Q_2 \dot \Phi_2 - \frac{1}{3C}\left(Q_1^2 - Q_1 Q_2 + Q_2^2\right) - \frac{1}{3L}\left(\Phi_1^2 - \Phi_1\Phi_2 + \Phi_2^2 \right)
\end{equation}
with 
\begin{equation}
  \begin{aligned}
    Q_1 &= q_{l_2} - q_{l_4} \\
    Q_2 &= q_{l_2} - q_{l_1}\\
    \Phi_1 &=  \phi_{v_1} - \phi_{v_3}\\
    \Phi_2 &= \phi_{v_4} - \phi_{v_1} 
  \end{aligned},
\end{equation}
and
\begin{equation}
  \{\Phi_j,Q_i\}  = \delta_{ij}.
\end{equation}
Therefore, we are free to write 
\begin{equation}
  H = \frac{1}{3 C}\left(Q_1^2 - Q_1Q_2 + Q_2^2\right)  + \frac{1}{3 L}\left(\Phi_1^2 - \Phi_1\Phi_2 + \Phi_2^2\right),
\end{equation}
with quantum commutation relations \begin{equation}
    [\Phi_j,Q_i] = \mathrm{i}\hbar \delta_{ij}.
\end{equation}

The process of ``integrating out" degrees of freedom associated to null vectors that arise during this process is conceptually straightforward. Still, as was the case in \cite{osborne2023symplectic}, nonlinear constraint equations may arise in the presence of nonlinear circuit elements without accompanying parasitic elements; these nonlinear equations may not have unique solutions, implying ambiguities in the correct phase space to quantize.   Since it was argued in \cite{rymarz_consistent_2022} that the addition of parasitic elements (which are present in any real experiment) qualitatively changes the spectrum, we do not focus on the technical question of picking the correct phase space in the presence of nonlinear constraints further in this manuscript.  

\section{Equivalence to other formulations}
Formally, it is true that the circuit Lagrangian (\ref{eqn:motherlag}) encodes all of the necessary physics: Kirchoff's rules are either manifestly obeyed, or correspond to the Euler-Lagrange equations of (\ref{eqn:motherlag}). It must then be true that the predictions of  (\ref{eqn:motherlag}) are equivalent to previous theories in the literature.  This equivalence is formally demonstrable, as we now summarize; see Appendix \ref{app:symmquant} for details and/or justifications of the claims made in this section.

\subsection{Branch--node formalism}\label{sec:bnf}
The most common approach to circuit quantization involves writing a Lagrangian in terms of fluxes that live on vertices alone: $\phi_v$. 
If we have linear capacitors, then one can directly integrate out $q_l$ in (\ref{eqn:motherlag}) to obtain a Lagrangian for $\phi_v$ alone with quadratic terms in $\dot \phi$.

However, it was recently pointed out in \cite{osborne2023symplectic} that more general circuits could be quantized by instead starting with a Lagrangian that depends on both $q_{e\in \mathcal{C}}$ (charges on capacitive edges) and $\phi_v$.  Comparing the flux-charge symmetric formulation to this one is slightly more subtle.
Recall that we may write
\begin{equation}
  M_{lv} = \sum_{e \in \mathcal C} B_{le} A_{ev}.
\end{equation}
Introduce a set of Lagrange multipliers $\lambda_e$ on capacitive edges, and write: 
\begin{equation}
      L' = \sum_{e \in \mathcal C}q_l B_{le} A_{ev}\dot\phi_v - \sum_{e \in \mathcal C} E_e^C(Q_e) - \sum_{e\in\mathcal I }E_e^L\left(\sum_v A_{ev}\phi_v\right) + \sum_{e\in C} \lambda_e \left(Q_e - \sum_l q_l B_{le}\right).
\end{equation}
The functional derivative 
\begin{equation}
    0 = \frac{\delta S}{\delta q_l} = \sum_{e \in \mathcal C} B_{le}\left[\sum_v A_{ev}\dot\phi_v - \lambda_e\right] 
\end{equation}
implies that the quantity 
\begin{equation}
    \gamma_e = \sum_v A_{ev}\dot\phi_v - \lambda_e
\end{equation}
forms a right null vector of the matrix $B$ restricted to the set of capacitive edges. 
Corollary \ref{cor:bnullv} shows that there is one such null vector for each capacitive cut of $G$. 
As a consequence, there exists $m$ vectors $|n_i\rangle = \sum_{e \in \mathcal C} \sigma_{e,i} |e\rangle$ with $\sigma_{e,i} \in \{-1,0,1\}$  with $i = 1,2,3,\dots,m$ such that 
\begin{equation}
    |\gamma\rangle = \sum_{i=1}^m \mu_i |n_i\rangle
\end{equation}
and thus  
\begin{equation}
    \lambda_e = \sum_v A_{ev}\dot\phi_v - \sum_{i = 1}^m \mu_i\sigma_{e,i}
\end{equation}
for some parameter $\mu$. 
Thus, it follows that $L'$ may be rewritten 
\begin{equation}\label{eqn:oldpaper}
    L' = \sum_{e \in \mathcal C} Q_e A_{ev}\dot\phi_v - \sum_{e \in \mathcal C} E_e^C(Q_e) - \sum_{e \in \mathcal I} \left(\sum_v A_{ev} \phi_v\right) - \sum_{i=1}^m\mu_i \sum_e \sigma_{e,i} Q_e.
\end{equation}
A functional derivative 
\begin{equation}
    0= \frac{\delta S }{\delta \mu_i} = \sum_{e}\sigma_{e,i} Q_e
\end{equation}
is easily understood as a consequence of 
\begin{equation}
    \sum_{l,e}q_l B_{le}\sigma_{e,i} = 0
\end{equation}
and can be recognized as a constraint that would follow from a Noether current in \cite{osborne2023symplectic}.  
Borrowing the terminology used in \cite{osborne2023symplectic}, the right null vectors of $B$ restricted to edges in $\mathcal C$ corresponded to a ``capacitively shunted island".   In the end, after using the Lagrange multiplier $\mu_i$ to fix the charges across capacitive cuts, we see that (\ref{eqn:oldpaper}) reproduces \cite{osborne2023symplectic}.



\subsection{Face--Edge formalism}

A possibly undesirable feature of the Lagrangian in (\ref{eqn:oldpaper}) is that capacitors are treated differently than inductors. With the flux-charge symmetric framework, however, we can produce a similarly universal theory of circuits that treats inductors as preferential. We write 
\begin{equation}
  M_{lv} = - \sum_{e \in \mathcal I} B_{le} A_{ev}
\end{equation}
In direct analogy with the discussion in Sec. \ref{sec:bnf}, we again introduce Lagrange multipliers $\lambda_e$, but this time they exist only on inductive edges:
\begin{equation}\label{eqn:hasslag}
      L' = -\sum_{e \in \mathcal I}q_l B_{le} A_{ev}\dot\phi_v - \sum_{e \in \mathcal C} E_e^C\left(\sum_l q_l B_{le}\right) - \sum_{e\in\mathcal I }E_e^L\left(\sum_v \psi_e \right) + \sum_{e\in \mathcal I} \lambda_e \left(\psi_e - \sum_v A_{ev}\phi_v\right).
\end{equation}
By taking functional derivatives and exploiting the properties of $A$, it is possible to perform a calulation very closely mirroring the one in Sec \ref{sec:bnf} to produce 
\begin{equation}\label{eqn:hasslagf}
    L' = -\sum_{e \in \mathcal I,l} q_l B_{le}\dot\psi_e 
    -\sum_{e \in \mathcal C} E_e^C(Q_e) - \sum_{e\in\mathcal I }E_e^L\left(\psi_e \right) .
\end{equation}

Of note, the role of matrix $A$ in (\ref{eqn:oldpaper}) is filled by $-B$ in (\ref{eqn:hasslagf}). From a formal perspective, the matrix $B$ is less well--behaved than $A$ since, for nonplanar graphs, $B$ may have less trivial null vectors than $A$, and there is no corresponding subtlety for the structure of $A$, even for nonplanar graphs. 
 (\ref{eqn:hasslagf}) is a so--called loop--branch\footnote{In the literature, it is sometimes said that any Lagrangian involving charges defined on a loop is called a ``loop--charge Lagrangian".} Lagrangian.

In the special case when the circuit corresponds to a planar graph, it was noted in \cite{ulrich_dual_2016} that Lagrangians of the form (\ref{eqn:hasslagf}) are dual to Lagrangians of the form (\ref{eqn:oldpaper}), as we discuss in Section \ref{sec:duality}. 
This fact is intimately related to the appearance of $B$ in (\ref{eqn:hasslagf}) instead of $A$.  Using the flux-charge symmetric Lagrangian (\ref{eqn:motherlag}), we will see a more straightforward derivation of such duality.



\section{Circuit duality}\label{sec:duality}
The fact that we are able to describe a circuit in either a framework that ``prefers" capacitors \emph{or} inductors is reminiscent of classical discussions of circuit duality \cite{bahar_generalized,Chua1974ExplicitTF,MacFarlane1969DualsystemMI,mathisham,Kerman2013}; see \cite{ulrich_dual_2016} for a discussion of such circuit dualities in the context of Lagrangian mechanics of superconducting circuits.   We now show that these dualities are especially transparent in our framework. 
A formal discussion of the results that follow is found in Appendix \ref{app:duality}. 

First we make a few general comments.  Any sensible notion of duality should be an involution (operation which is the identity when applied twice) on classical phase space (or quantum Hilbert space) of a circuit.  A circuit is said to be self-dual if its Hamiltonian is invariant under this transformation.

One kind of duality arises by a mere ``relabeling transformation" on (\ref{eqn:motherlag}).  With the Lagrangian (\ref{eqn:motherlag}) in mind, consider the transformation given by
\begin{equation}\label{eqn:dualitytrans}
  \begin{aligned}
    \phi_v &\rightarrow \phi_v^* = q'_v \\
    q_l &\rightarrow q_l^* = \phi_l' \\ 
    A &\rightarrow A^* = B^{\mathrm T} \\ 
    B &\rightarrow B^* = A^{\mathrm T} \\ 
    \mathcal C &\rightarrow \mathcal C^* \simeq \mathcal I \\ 
    \mathcal I &\rightarrow \mathcal I^* \simeq \mathcal C \\ 
    \mathcal V &\rightarrow \mathcal V^* \simeq \mathcal L \\ 
    \mathcal F &\rightarrow \mathcal L^* \simeq \mathcal V.
  \end{aligned}
\end{equation}
The table in Fig.\ref{fig:dualityfig} illustrates the transformation (\ref{eqn:dualitytrans}). 
Informally, (\ref{eqn:dualitytrans}) swaps fluxes with charges, capacitors with inductors, and nodes with faces. 
Under (\ref{eqn:dualitytrans}), the connection matrix of $G$ transforms as\footnote{We remark that $X^*$ is the dual of some object $X$ under the transformation (\ref{eqn:dualitytrans}) and not, for example, complex conjugation. } 
\begin{equation}
  M \rightarrow M^* = - M^{\mathrm T}. 
\end{equation}
Therefore, the Lagrangian in (\ref{eqn:motherlag}) transforms as 
\begin{equation}
 L \rightarrow L^* = 
 \sum_{l,v}  q_l^* M_{fv}^* \dot\phi_v^* - \sum_{e \in \mathcal C^*} E_e^C\left(\sum_{l} q_l^* B_{fe}^*\right) - \sum_{e \in \mathcal I^*}E_e^L\left(\sum_{v} A_{ev}^* \phi_v^*\right)
\end{equation}
or 
\begin{equation}
  L^* = \sum_{l,v} q'_v  M_{vl} \dot \phi'_l   - \sum_{e \in \mathcal I} E_e^C \left(\sum_{l \in \mathcal V^*} A_{el}\phi'_l\right)  - \sum_{e \in \mathcal C} E_e^L\left(\sum_{v \in \mathcal F^*} q'_v B_{ve}\right)
\end{equation}
where we have integrated the first term by parts after carrying out the substitutions in (\ref{eqn:dualitytrans}).


\begin{figure}
    \centering
    \begin{tabular}{|c|c|}
    \hline
       object  & dual object \\
       \hline
        \begin{circuitikz}[scale=1.5] \draw[thick] (0,0.5) to[inductor,i_=$e$] (1,0.5);
        \filldraw[black] (0,0.5) circle (1.5pt) node[anchor=east]{$v_1$};
        \filldraw[black] (1,0.5) circle (1.5pt) node[anchor=west]{$v_2$};
        \draw[thin, <-] (0.5,0)node{$l_1$}  ++(-60:0.2) arc (-60:170:0.2);
        \draw[thin, <-] (0.5,1)node{$l_2$}  ++(-60:0.2) arc (-60:170:0.2);
        \end{circuitikz} &  
        \begin{circuitikz}[scale=1.5]\draw[thick] (0.5,0) to[capacitor,i_=$e^*$] (0.5,1);
        \filldraw[black] (0.5,0) circle (1.5pt) node[anchor=north]{$l_1^*$};
        \filldraw[black] (0.5,1) circle (1.5pt) node[anchor=south]{$l_2^*$};
        \draw[thin, <-] (0,0.5)node{$v_1^*$}  ++(170:0.2) arc (170:-60:0.2);
        \draw[thin, <-] (1,0.5)node{$v_2^*$}  ++(170:0.2) arc (170:-60:0.2);
        \end{circuitikz}  \\ 
        \hline
        \begin{circuitikz}[scale=1.5] 
        \draw[thick] (0,0.5) to[capacitor,i_=$e$] (1,0.5);
        \filldraw[black] (0,0.5) circle (1.5pt) node[anchor=east]{$v_1$};
        \filldraw[black] (1,0.5) circle (1.5pt) node[anchor=west]{$v_2$};
        \draw[thin, <-] (0.5,0)node{$l_1$}  ++(-60:0.2) arc (-60:170:0.2);
        \draw[thin, <-] (0.5,1)node{$l_2$}  ++(-60:0.2) arc (-60:170:0.2);
        \end{circuitikz} &  \begin{circuitikz}[scale=1.5]\draw[thick] (0.5,0) to[inductor,i_=$e^*$] (0.5,1);
        \filldraw[black] (0.5,0) circle (1.5pt) node[anchor=north]{$l_1^*$};
        \filldraw[black] (0.5,1) circle (1.5pt) node[anchor=south]{$l_2^*$};
        \draw[thin, <-] (0,0.5)node{$v_1^*$}  ++(170:0.2) arc (170:-60:0.2);
        \draw[thin, <-] (1,0.5)node{$v_2^*$}  ++(170:0.2) arc (170:-60:0.2);
        \end{circuitikz}  \\ 
        \hline
        $A_{e v_1} = -1$ &  $B^*_{v_1^* e^*} = -1 $ \\ 
        \hline
        $B_{l_1e} = 1$   & $A^*_{e^* l_1^*} = 1$ \\ 
        \hline
        $M_{l_1 v_1}$ & $- M_{v_1^*, l_1^*} $ \\ 
        \hline
        $\phi_{v_1}$ & $q_{v_1^*}$ \\ 
        \hline
        $q_{l_1}$ & $\phi_{l_1^*}$\\
        \hline
    \end{tabular}
    \caption{A number of examples showing how the transformation (\ref{eqn:dualitytrans}) effects various aspects of a circuit. We emphasize that, from the perspective of this formalism, nonlinear inductors may be treated simply as inductors on equal footing with linear inductors. In particular, a Josephson junction would be dual to a quantum phase slip.  }
    \label{fig:dualityfig}
\end{figure}

The non-trivial question is whether $L^*$ can be interpreted as the Lagrangian for an \emph{actual circuit}, where $q^*_l$ represent physical loop charges and $\phi^*_v$ represent physical node fluxes.   We address this for the remainder of the section.

\subsection{Planar circuits}
\begin{figure}
    \centering
    \begin{subfigure}{0.49\linewidth}
    \scalebox{0.7}{
        \begin{circuitikz}
        \draw[thick] (0,0) to[capacitor,i_=$e_5$] (3,0) to (4,0) to[inductor,i_=$e_1$] (4,4) to[capacitor,i_=$e_6$] (0,4) to[inductor,i_=$e_2$] (0,0);
        \draw[thick] (0,0) to[inductor,i_=$e_3$] (2,2) to[capacitor,i_=$e_8$] (0,4);
        \draw[thick] (4,0) to[inductor,i_=$e_4$] (6,2) to[capacitor,i_=$e_7$] (4,4);
        \filldraw[black] (0,0) circle (2pt) node[anchor = east]{$v_1$};
        \filldraw[black] (4,0) circle (2pt) node[anchor = north]{$v_2$}; 
        \filldraw[black] (4,4) circle (2pt) node[anchor = south]{$v_3$}; 
        \filldraw[black] (0,4) circle (2pt) node[anchor = east]{$v_4$};
        \filldraw[black] (2,2) circle (2pt) node[anchor = south west]{$v_5$};
        \filldraw[black] (6,2) circle (2pt) node[anchor = west]{$v_6$}; 

        \filldraw[gray,opacity = 0.5] (1,2) circle (2pt);
        \filldraw[gray,opacity = 0.5] (3,2) circle (2pt);
        \filldraw[gray,opacity = 0.5] (5,2) circle (2pt);
        \filldraw[gray,opacity = 0.5] (-1,2) circle (2pt);
        \draw[thick,gray,opacity = 0.5] (1,2) to (0,2) to[capacitor] (-1,2); 
        \draw[thick,gray,opacity = 0.5] (1,2) to (2,3) to[inductor] (3,2); 
        \draw[thick,gray,opacity = 0.5] (1,2) to (2,1) to[capacitor] (3,2); 
        \draw[thick,gray,opacity = 0.5] (3,2) to (3,4.5) to (-1,4.5) to[inductor] (-1,2); 
        \draw[thick,gray,opacity = 0.5] (3,2) to (3,-.5) to (-1,-.5) to[inductor] (-1,2); 
        \draw[thick,gray,opacity = 0.5] (3,2) to (4,2) to[capacitor] (5,2); 
        \draw[thick,gray,opacity = 0.5] (5,2) to (5,5) to (-2,5) to[inductor] (-1,2); 
        \draw[thick,gray,opacity = 0.5] (5,2) to (5,-1) to (-2,-1) to[capacitor] (-1,2); 
        
        \draw[thin, <-] (.9,2)node{$l_1$}  ++(-60:0.5) arc (-60:170:0.5);
        \draw[thin, <-] (2.9,2)node{$l_2$}  ++(-60:0.5) arc (-60:170:0.5);
        \draw[thin, <-] (4.9,2)node{$l_3$}  ++(-60:0.5) arc (-60:170:0.5);
        \draw[thin, <-] (-1,2)node{$l_4$}  ++(60:0.5) arc (60:-170:0.5);
        \end{circuitikz}}
        \caption{}
    \end{subfigure}
    \begin{subfigure}{0.49\linewidth}
    \scalebox{0.7}{
        \begin{circuitikz}
        \draw[thick] (0,0) to[capacitor,i_=$e_5$] (3,0) to (4,0) to[inductor,i_=$e_1$] (4,4) to[capacitor,i_=$e_6$] (0,4) to[inductor,i_=$e_2$] (0,0);
        \draw[thick] (0,0) to[inductor,i_=$e_3$] (-2,2) to[capacitor,i_=$e_8$] (0,4);
        \draw[thick] (4,0) to[inductor,i_=$e_4$] (6,2) to[capacitor,i_=$e_7$] (4,4);
        \filldraw[black] (0,0) circle (2pt) node[anchor = east]{$v_1$};
        \filldraw[black] (4,0) circle (2pt) node[anchor = north]{$v_2$}; 
        \filldraw[black] (4,4) circle (2pt) node[anchor = south]{$v_3$}; 
        \filldraw[black] (0,4) circle (2pt) node[anchor = east]{$v_4$};
        \filldraw[black] (-2,2) circle (2pt) node[anchor = east]{$v_5$};
        \filldraw[black] (6,2) circle (2pt) node[anchor = west]{$v_6$}; 
            
        \filldraw[gray,opacity = 0.5] (-1,2) circle (2pt);
        \filldraw[gray,opacity = 0.5] (2,2) circle (2pt);
        \filldraw[gray,opacity = 0.5] (5,2) circle (2pt);
        \filldraw[gray,opacity = 0.5] (8,2) circle (2pt);
        \draw[thick,gray,opacity = 0.5] (2,2) to[capacitor] (0,2) to (-1,2);
        \draw[thick,gray,opacity = 0.5] (2,2) to[capacitor] (4,2) to (5,2);
        \draw[thick,gray,opacity = 0.5] (2,2) to[inductor] (2,4) to (2,4.4) to (5.5,4.4) to (8,2);
        \draw[thick,gray,opacity = 0.5] (2,2) to[inductor] (2,0) to (2,-0.4) to (5.5,-.4) to (8,2);
        \draw[thick,gray,opacity = 0.5] (-1,2) to (-1,4.6) to (8,4.6) to[inductor] (8,2);
        \draw[thick,gray,opacity = 0.5] (-1,2) to (-1,-0.6) to (8,-.6) to[capacitor] (8,2);
        \draw[thick,gray,opacity = 0.5] (5,2) to (6,1.5) to[capacitor] (8,2);
        \draw[thick,gray,opacity = 0.5] (5,2) to (6,2.5) to[inductor] (8,2);
        
        \draw[thin, <-] (-1,2)node{$l_1$}  ++(-60:0.5) arc (-60:170:0.5);
        \draw[thin, <-] (2,2)node{$l_2$}  ++(-60:0.5) arc (-60:170:0.5);
        \draw[thin, <-] (4.9,2)node{$l_3$}  ++(-60:0.5) arc (-60:170:0.5);
        \draw[thin, <-] (-3,2)node{$l_4$}  ++(60:0.5) arc (60:-170:0.5);
        \end{circuitikz}}
        \caption{}
    \end{subfigure}
    \caption{The black circuits in (a) and (b) are simply redrawings of the same circuit. The duals of the circuits (a) and (b) are drawn in grey.  In both subfigures, labels correspond to the circuit in black.}
    \label{fig:ulrich}
\end{figure}
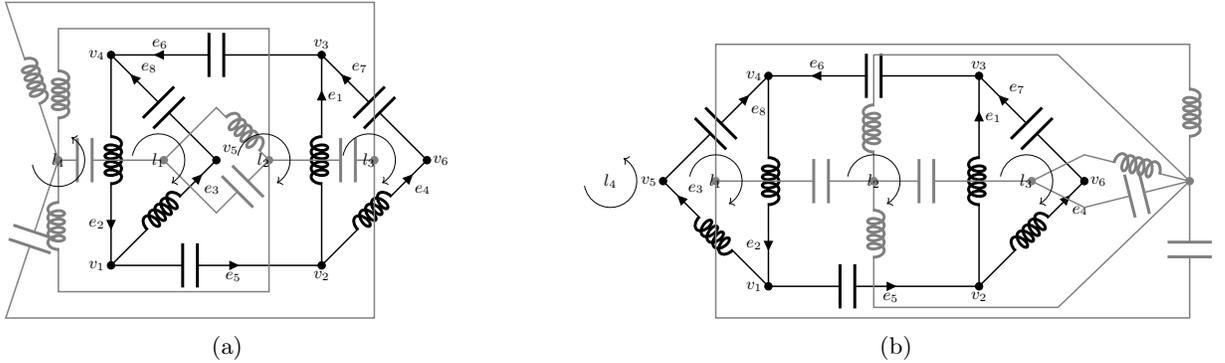

One way to understand the subtlety of this duality transformation is to ask whether it behaves ``nicely" on the \emph{drawings} of a circuit.  In other words, does $L^*$ describe a physical circuit that we can easily draw? 
In what follows, we will assume that capacitances and inductances have the same units, or equivalently that charges and fluxes have the same units. This is a matter of convenience that can be adapted to a physical system by making suitable variable redefinitions. 
The discussion in this section only applies to planar circuits, and in the next section we will discuss the problems that may arise for nonplanar circuits. 

The procedure for constructing the dual of a given circuit $G$ is as follows \cite{ulrich_dual_2016,gross_topological}:
\begin{enumerate}
    \item Draw $G$ on a plane (or equivalently a sphere, by identifying points at infinite distance on the plane with a single point). 
    \item For every face\footnote{Recall that there are no loops that cannot be drawn as faces for a planar graph.} $l$ in $G$, draw a node inside the $l$ labeled $l^*$
    \item Every branch $e$ in $G$ lies on the boundary of exactly two faces, say $l_1$ and $l_2$. If $e$ is capacitive (inductive), draw an inductive (capacitive) branch labeled $e^*$ connecting $l_1^*$ and $l_2^*$.
\end{enumerate}
An example of this drawing procedure is given in Fig.\ref{fig:ulrich}.

We remark that $B$, the orientation matrix of $G$, depends upon how $G$ is drawn. Since the incidence matrix of the dual of $G$ is the transpose of $B$, it follows that a single circuit can have multiple dual circuits. 
Since $G$'s incidence matrix is uniquely defined, all duals of $G$ have the same orientation matrix, and the number of duals of $G$ with distinct incidence matrices is exactly the number of drawings of $G$ with different orientation matrices. 
For example, in Fig.\ref{fig:ulrich}(b), the dual of $G$  has a node of degree six, which the dual of $G$ in (a) has no such vertex. In this way, we can see that circuits need not have a unique dual. 
A detailed discussion of this circuit is in Section \ref{ex:multi}.

The manner in which incidence and orientation matrices are ``exchanged" under duality can be interpreted in the following way: the combinatorial information of a graph is encoded in the topological information of its dual, and the topological information of a graph determines the combinatorial information in its dual. In this way, it is sensible to say that, for graphs (and thus circuits), topology is dual to combinatorics. 

\subsection{Nonplanar circuits}
The transformation (\ref{eqn:dualitytrans}) is the most natural  candidate for a duality transformation at the Lagrangian level.
Unfortunately, as we now discuss, for nonplanar circuits it is not general clear how to construct a physical circuit for which $L^*$ is its Lagrangian.
The difficulty in constructing the dual to a nonplanar circuit arises when attempting to draw the dual circuit, rather than when trying to produce the dual Lagrangian. 

\begin{figure}
    \centering
    \begin{circuitikz}
    \draw[thick,red] (0,0) to[capacitor,i^=$e_1$,color=red] (0,2) to[capacitor,i^=$e_2$,color=red]  (0,4) to[capacitor,i^=$e_3$,color=red]  (0,6) ;
    \draw[thick,blue] (0,6) to[capacitor,i^=$e_4$,color=blue]  (2,6) to[capacitor,i^=$e_5$,color=blue]  (4,6) to[capacitor,i^=$e_6$,color=blue]  (6,6);
    \draw[thick,blue] (0,0) to[capacitor,i_=$e_4$,color=blue]  (2,0) to[capacitor,i_=$e_5$,color=blue]  (4,0) to[capacitor,i_=$e_6$,color=blue]  (6,0);
    \draw[thick,red] (6,0) to[capacitor,i_=$e_1$,color=red] (6,2) to[capacitor,i_=$e_2$,color=red]  (6,4) to[capacitor,i_=$e_3$,color=red]  (6,6);
    \draw[thick] (4,6) to[inductor,i^=$e_7$] (0,4) to[inductor,i^=$e_8$] (2,0) to[inductor,i^=$e_9$] (6,2) to[inductor,i^=$e_{10}$] (4,6);
    \filldraw[black] (0,0) circle (2pt) node[anchor=north]{$v_1$};
    \filldraw[black] (2,0) circle (2pt) node[anchor=north]{$v_2$};
    \filldraw[black] (4,0) circle (2pt) node[anchor=north]{$v_3$};
    \filldraw[black] (6,0) circle (2pt) node[anchor=north west]{$v_1$};
    \filldraw[black] (6,2) circle (2pt) node[anchor=west]{$v_4$};
    \filldraw[black] (6,4) circle (2pt) node[anchor=west]{$v_5$};
    \filldraw[black] (6,6) circle (2pt) node[anchor=south west]{$v_1$};
    \filldraw[black] (4,6) circle (2pt) node[anchor=south]{$v_3$};
    \filldraw[black] (2,6) circle (2pt) node[anchor=south]{$v_2$};
    \filldraw[black] (0,6) circle (2pt) node[anchor= south east]{$v_1$};
    \filldraw[black] (0,4) circle (2pt) node[anchor=east]{$v_5$};
    \filldraw[black] (0,2) circle (2pt) node[anchor=east]{$v_4$}; 
    
    \draw[thin, <-] (3,3)node{$l_1$}  ++(-60:0.3) arc (-60:170:0.3);
    \draw[thin, <-] (1,1)node{$l_2$}  ++(-60:0.3) arc (-60:170:0.3);
    \draw[thin, <-] (5,1)node{$l_3$}  ++(-60:0.3) arc (-60:170:0.3);
    \draw[thin, <-] (1,5)node{$l_4$}  ++(-60:0.3) arc (-60:170:0.3);
    \draw[thin, <-] (5,5)node{$l_5$}  ++(-60:0.3) arc (-60:170:0.3);
    
    \end{circuitikz}
    \caption{A circuit on $K_5$. The edges with alike labels are to be identified. The loops $l_6$ and $l_7$ are drawn in red and blue respectively.}
    \label{fig:nonplanar}
\end{figure}
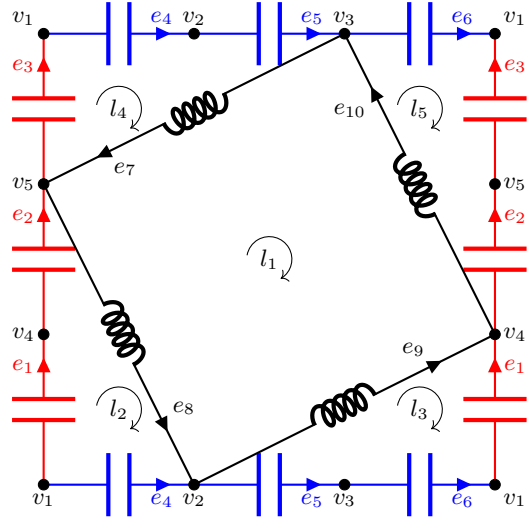

For planar graphs, it is possible to choose a suitable set of loops $\mathcal L$ by drawing the circuit in question on a surface (namely the plane), and such circuits can be made to have the property that 
\begin{equation}
    \sum_{l} B_{le} = 0.
\end{equation}
The importance of this property is that every edge in the circuit participates in precisely two loops within $\mathcal L $.
On the other hand, for all circuits, including nonplanar ones, we demand 
\begin{equation}
    \sum_v A_{ev} = 0.
\end{equation}
The duality transformation (\ref{eqn:dualitytrans}) sends 
\begin{equation}
    \begin{aligned}
        B &\rightarrow B^* = A^{\mathrm T} \\ 
        A &\rightarrow A^* = B^{\mathrm T}. 
    \end{aligned}
\end{equation}
If 
\begin{equation}
    \sum_{l \in \mathcal L} B_{le} \neq 0,
\end{equation}
then $A^* = B^{\mathrm T}$ is \emph{not a valid incidence matrix}. 
The reason for this is the $2g$ ``topological cycles" that must be included in $\mathcal L$ in order to fully specify all of the currents in the nonplanar circuit of interest. 
See Fig \ref{fig:nonplanar} for a drawing of a circuit with topological cycles chosen and emphasized with color coding.
Consider some edge $e$ that participates in a topological cycle $l_t$. $e$ also lies on the boundary of a pair of faces, say $l_1$ and $l_2$. 
Then, 
\begin{equation}
    \sum_l B_{le} = B_{l_t e} + B_{l_1 e} + B_{l_2 e} = B_{l_t e} \neq 0.
\end{equation}
Thus, if one attempted to interpret $B^{\mathrm T}$ as an incidence matrix, the edge $e$ would ``connect" three vertices, rather than two, but then the resulting construction is not an edge. 
From a mathematical perspective, it is no issue to acknowledge that an edge in a graph participates in many cycles, but there is no sensible notion of an edge that connects many vertices.  

Moreover, any notion of graph duality must conserve $g$, the genus of the surface of embedding. We see, however that 
\begin{equation}
    B \rightarrow B^* = A^{\mathrm T}
\end{equation}
gives rise to $B^*$ that \emph{does} have the property that 
\begin{equation}\label{eqn:planarcond}
    \sum_l B^*_{le} = 0, 
\end{equation}
which is a hallmark property of planar graphs in the sense that (\ref{eqn:planarcond}) holds for and only for planar graphs. 

Hence, there exist nonplanar circuits with no physical dual. We remark that (\ref{eqn:dualitytrans}) is still perfectly well--defined for nonplanar circuits, but that it may not be possible to draw a circuit that gives rise to the resulting Hamiltonian.   It may be possible that one can always introduce auxiliary degrees of freedom and integrate them out, such that $L^*$ is the reduced description of a circuit with (nonlinear) capacitive loops or inductive cuts, but we have not found an elegant and fully general construction which allows us to construct an honest circuit dual for a general nonplanar circuit. Resolving this question seems to us to be an important open problem.

\section{Examples}
In what follows, we will solve a number of examples that are either particularly pedagogical or of particular interest. 
\subsection{A planar circuit}\label{ex:multi}
Consider the black circuit drawn in Fig.\ref{fig:ulrich}(a), which we will call $G$. The face in the dual of $G$ (drawn in grey) corresponding to the vertex $v_i$ in $G$ will be labeled $v_i^*$. We will refer to the grey circuit as $G^*$.
The capacitance (inductance) of the circuit on element on $e_i$ will be denoted as $C_i$ ($L_i$). 
We suppose for the sake of simplicity in the presentation that all circuit elements are linear, but we emphasize that this is unnecessary. 
We will begin by analyzing $G$. 

To start, 
\begin{equation}
A = 
    \begin{blockarray}{ccccccc}
    v_1 & v_2 & v_3 & v_4 & v_5 & v_6 & \\
    \begin{block}{(cccccc)@{\hspace{7pt}}c}
        0 & -1 & 1 & 0 & 0 & 0 & e_1\\ 
        1 & 0 & 0 & -1 & 0 & 0& e_2 \\ 
        -1 & 0 & 0 & 0 & 1 & 0& e_3 \\ 
        0 & -1 & 0 & 0 & 0 & 1& e_4 \\ 
        -1 & 1 & 0 & 0 & 0 & 0& e_5 \\ 
        0 & 0 & -1 & 1 & 0 & 0& e_6 \\ 
        0 & 0 & 1 & 0 & 0 & -1& e_7 \\ 
        0 & 0 & 0 & 1 & -1 & 0& e_8 \\ 
        \end{block}
        \end{blockarray}
\end{equation}
and 
\begin{equation}
B = 
    \begin{blockarray}{ccccccccc} 
    e_1 & e_2 & e_3 & e_4 & e_5 & e_6 & e_7 & e_8 & \\ 
        \begin{block}{(cccccccc)@{\hspace{7pt}}c}
        0 & -1 & -1 & 0 & 0 & 0 & 0 & -1 & l_1 \\  
        -1 & 0 & 1 & 0 & -1 & -1 & 0 & 1 & l_2 \\ 
        1 & 0 & 0 & -1 & 0 & 0 & -1 & 0 & l_3 \\ 
        0 & 1 & 0 & 1 & 1 & 1 & 1 & 0& l_4  \\ 
        \end{block}
    \end{blockarray}
\end{equation}
and thus 
\begin{equation}
    M = 
\begin{blockarray}{ccccccc}
v_1 & v_2 & v_3 & v_4 & v_5 & v_6 & \\ 
\begin{block}{(cccccc)@{\hspace{7pt}}c}
 0 & 0 & 0 & -1 & 1 & 0 & l_1\\
 1 & -1 & 1 & 0 & -1 & 0& l_2 \\
 0 & 0 & -1 & 0 & 0 & 1 & l_3 \\
 -1 & 1 & 0 & 1 & 0 & -1& l_4 \\
 \end{block}
 \end{blockarray},
\end{equation}
Now, the Lagrangian for $G$ is given by 
\begin{equation}
    \begin{aligned}
        L &= (q_4 - q_2)(\dot\phi_{v_2} - \dot\phi_{v_1}) + (q_2 - q_1)(\dot\phi_{v_4} - \dot\phi_{v_5}) + (q_4 - q_3)(\dot\phi_{v_3} - \dot\phi_{v_6}) + (q_4 - q_2) (\dot\phi_{v_4} - \dot\phi_{v_3}) \\
        &-\frac{1}{2 L_1} (\phi_{v_3} - \phi_{v_2})^2-\frac{1}{2 L_2} (\phi_{v_1} - \phi_{v_4})^2 -\frac{1}{2 L_3} (\phi_{v_5} - \phi_{v_1})^2 -\frac{1}{2 L_4} (\phi_{v_6} - \phi_{v_2})^2 \\ 
        &-\frac{1}{2 C_5}(q_4 - q_2)^2 - \frac{1}{2 C_6} (q_4 - q_2)^2 - \frac{1}{2 C_7}(q_4 - q_3)^2 - \frac{1}{2 C_8}(q_2 - q_1)^2 .
    \end{aligned}
\end{equation}
There are two right null vectors of $M$. The first corresponds to 
\begin{equation}
    0 = \frac{\delta S}{\delta \phi_{v_2}} + \frac{\delta S }{\delta \phi_{v_3}} + \frac{\delta S}{\delta \phi_{v_6}}
\end{equation}
which contains no nontrivial constraint. 
The other nontrivial null vector corresponds to 
\begin{equation}
    0 = \frac{\delta S }{\delta \phi_{v_1}} + \frac{\delta S}{\delta \phi_{v_2}} = \frac{1}{L_1}(\phi_{v_2} - \phi_{v_3}) + \frac{1}{L_2} (\phi_{v_1} - \phi_{v_4}) + \frac{1}{L_4}(\phi_{v_1} - \phi_{v_5}) + \frac{1}{L_4}(\phi_{v_2} - \phi_{v_6}).
\end{equation}
Resolving this constraint and defining 
\begin{equation}
    \begin{aligned}
        Q_1 &= q_4 - q_2 \\ 
        Q_2 &= q_2 - q_1 \\ 
        Q_3 &= q_4 - q_3 \\ 
        \Phi_1 &= \phi_{v_2} + \phi_{v_4} - \phi_{v_1} - \phi_{v_3} \\ 
        \Phi_2 &= \phi_{v_4} - \phi_{v_5} \\ 
        \Phi_3 &= \phi_{v_3} - \phi_{v_6},
    \end{aligned}
\end{equation}
we find immediately 
\begin{equation}
\begin{aligned}
L & = Q_1 \dot \Phi_1 + Q_2 \dot \Phi_2 + Q_3 \dot \Phi_3  - \frac{1}{2}\left(\frac{1}{C_5} + \frac{1}{C_6}\right)Q_1^2 - \frac{1}{2 C_7}Q_3^2 - \frac{1}{2 C_8}Q_2^2 \\ 
    &-\frac{L_{\Sigma}}{2}\left(\frac{1}{L_1L_2}\Phi _1^2+ \frac{1}{L_1L_3} \left(\Phi _1-\Phi _2\right){}^2+\frac{1}{L_1L_4}\Phi _3^2+\frac{1}{L_2L_3} \Phi _2^2+\frac{1}{L_2L_4} \left(\Phi _1+\Phi _3\right)^2+\frac{1}{L_3L_4} \left(\Phi _1-\Phi _2+\Phi _3\right)^2\right) \\ 
    \end{aligned}
\end{equation}
with 
\begin{equation}
    \frac{1}{L_\Sigma} = \frac{1}{L_1} + \frac{1}{L_2} + \frac{1}{L_3} + \frac{1}{L_4}. 
\end{equation}
By taking the transformation 
\begin{equation}
    \begin{aligned}
        Q_i &\rightarrow Q_i^* = -\Phi_i \\ 
        \Phi_i &\rightarrow \Phi_i^* = Q_i,
    \end{aligned}
\end{equation}
we can acquire the Lagrangian (and Hamiltonian) of the grey circuit drawn in Fig \ref{fig:ulrich} (a).

We remark that the black circuit drawn in Fig \ref{fig:ulrich} (b) is simply a different drawing of the black circuit drawn in Fig \ref{fig:ulrich}(a), but the grey circuit in Fig \ref{fig:ulrich} (b) is not even graph isomorphic to the grey circuit drawn in Fig \ref{fig:ulrich} (a). A similarly simple analysis of the black circuit drawn in Fig \ref{fig:ulrich} (b) would provide an canonically related Hamiltonian to the one governing $G$  (of course, since the two circuits are identical). It then follows that the Hamiltonian governing the grey circuits must also be canonically related. 


\subsection{A self--dual circuit}
Consider the circuit drawn in Fig.\ref{fig:excir}. The Hamiltonian produced is given by 
\begin{equation}
    H(Q, \Phi) = \frac{1}{3 C} (Q_1^2 - Q_1 Q_2 + Q_2^2) + \frac{1}{3 L}(\Phi_1^2 - \Phi_1 \Phi_2 + \Phi_2^2).
\end{equation}
It is obvious that $H(-\Phi,Q) = H(Q,\Phi)$ provided that we choose units where $C = L$, which is to say that the circuit drawn in Fig.\ref{fig:excir} is, or at least \emph{has} a self--dual drawing.  
%

\subsection{A circuit on $K_5$}
For the purposes of this example, we will focus on the circuit drawn in Fig.~\ref{fig:nonplanar}, which will again be called $G$. The drawing at hand is slightly unconventional since it is drawn on a polygon with boundaries identified (which is typical in mathematical descriptions of a torus, for example). For the sake of simplicity, we will suppose that all inductances take on a value $L$ and all capacitances take on a value $C$. 

First, we must compute the relevant $A$, $B$, and $M$ matrices: 
\begin{equation}
\begin{aligned}
    A &= 
    \begin{blockarray}{cccccc}
     v_1 & v_2 & v_3 & v_4 & v_5 & \\ 
        \begin{block}{(ccccc)@{\hspace{7pt}}c}
 -1 & 0 & 0 & 1 & 0 & e_1 \\
0 & 0 & 0 & -1 & 1& e_2  \\
1 & 0 & 0 & 0 & -1& e_3  \\
-1 & 1 & 0 & 0 & 0& e_4  \\
0 & -1 & 1 & 0 & 0& e_5  \\
1 & 0 & -1 & 0 & 0& e_6  \\
0 & 0 & -1 & 0 & 1& e_7  \\
0 & 1 & 0 & 0 & -1& e_8  \\
0 & -1 & 0 & 1 & 0& e_9  \\
0 & 0 & 1 & -1 & 0& e_{10}  \\
        \end{block}
    \end{blockarray}, \\ 
    B &= 
    \begin{blockarray}{ccccccccccc}
    e_1 & e_2 & e_3 & e_4 & e_5 & e_6 & e_7 & e_8 & e_9 & e_{10} & \\
        \begin{block}{(cccccccccc)@{\hspace{7pt}}c}
 0 & 0 & 0 & 0 & 0 & 0 & -1 & -1 & -1 & -1 & l_1 \\
 1 & 1 & 0 & -1 & 0 & 0 & 0 & 1 & 0 & 0 & l_2 \\
 -1 & 0 & 0 & 0 & -1 & -1 & 0 & 0 & 1 & 0  & l_3\\
 0 & 0 & 1 & 1 & 1 & 0 & 1 & 0 & 0 & 0  & l_4\\
 0 & -1 & -1 & 0 & 0 & 1 & 0 & 0 & 0 & 1  & l_5\\
 1 & 1 & 1 & 0 & 0 & 0 & 0 & 0 & 0 & 0  & l_6\\
 0 & 0 & 0 & 1 & 1 & 1 & 0 & 0 & 0 & 0  & l_7\\
        \end{block}
    \end{blockarray}\\ 
    M &=
    \begin{blockarray}{cccccc}
     v_1 & v_2 & v_3 & v_4 & v_5 & \\ 
        \begin{block}{(ccccc)@{\hspace{7pt}}c}
 0 & 0 & 0 & 0 & 0 & l_1\\
 0 & -1 & 0 & 0 & 1& l_2 \\
 0 & 1 & 0 & -1 & 0& l_3 \\
 0 & 0 & 1 & 0 & -1& l_4 \\
 0 & 0 & -1 & 1 & 0& l_5 \\
 0 & 0 & 0 & 0 & 0 & l_6 \\
 0 & 0 & 0 & 0 & 0 & l_7\\
        \end{block}
    \end{blockarray}
    \end{aligned}
\end{equation}
Thus, 
\begin{equation}
\begin{aligned}
    L &=\left(q_{l_3}-q_{l_2}\right) (\dot\phi _{v_2}- \dot\phi_{v_5})+\left(q_{l_4}-q_{l_5}\right) (\dot\phi_{v_3} - \dot\phi_{v_5})+\left(q_{l_5}-q_{l_3}\right) (\dot\phi_{v_4} - \dot\phi_{v_5})\\
    &- \frac{1}{2 C} \left( (q_{l_2} - q_{l_3} + q_{l_6})^2 + (q_{l_2} - q_{l_5} + q_{l_6})^2 + (q_{l_4} - q_{l_5}+ q_{l_6})^2\right) \\
    &-\frac{1}{2C}\left( (q_{l_4} - q_{l_2} + q_{l_7})^2 + (q_{l_4} - q_{l_3} + q_{l_7})^2 + (q_{l_5} - q_{l_3} + q_{l_7})^2\right) \\ 
    &- \frac{1}{2 L} \left((\phi_{v_5} -\phi_{v_3})^2 + (\phi_{v_2} - \phi_{v_5})^2 + (\phi_{v_4} - \phi_{v_2})^2 + (\phi_{v_4} - \phi_{v_3})^2\right).
    \end{aligned}
\end{equation}
Choosing the variables 
\begin{equation}
    \begin{aligned}
        Q_1 = q_{l_3} - q_{l_2} \\ 
        Q_2 = q_{l_4} - q_{l_5} \\ 
        Q_3 = q_{l_5} - q_{l_3} \\
        \Phi_1 = \phi_{v_2} -\phi_{v_5}\\ 
        \Phi_2 = \phi_{v_3} -\phi_{v_5}\\ 
        \Phi_3 = \phi_{v_4} -\phi_{v_5}
    \end{aligned}
\end{equation}
we can rewrite
\begin{equation}\label{eqn:nonplanarL}
\begin{aligned}
    L = \sum_{i=1}^3 Q_i\dot\Phi_i - \frac{1}{2 C} \left(
    (q_{l_6} - Q_1)^2 + ( q_{l_6} - Q_1 - Q_3)^2 + (Q_2 + q_{l_6})^2 
    \right) \\ 
    - \frac{1}{2 C}\left((q_{l_7} + Q_1 + Q_2 + Q_3)^2  + (q_{l_7}+  Q_2 + Q_3)^2 + (q_{l_7}+  Q_3)^2\right) \\ 
    - \frac{1}{2 L}\left(\Phi_2^2 + \Phi_1^2 + (\Phi_3 - \Phi_1)^2 + (\Phi_2 - \Phi_3)^2\right).
    \end{aligned}
\end{equation}
From (\ref{eqn:nonplanarL}) it is clear to see that $q_{l_6}$ and $q_{l_7}$ are constrained in terms of $Q$ variables as a result of 
\begin{equation}\label{eqn:nonplanrcons}
        0 = \frac{\delta S}{\delta q_{l_6}}
         = \frac{\delta S}{\delta q_{l_7}}.
\end{equation}
After using (\ref{eqn:nonplanrcons}), we find 
\begin{equation}
\begin{aligned}
    q_{l_6} &= \frac{1}{3} (2Q_1- Q_2 + Q_3)  \\ 
    q_{l_7} &= \frac{1}{3}(-Q_1 - 2 Q_2 - 3 Q_3).
\end{aligned}
\end{equation}
Then, the Hamiltonian describing $G$ is given by $H(Q_1,Q_2,Q_3,\Phi_1,\Phi_2,\Phi_3) = \sum_{i=1}^3 Q_i \dot\Phi_i - L$. 
The Poisson brackets of the system at hand are 
\begin{equation}
    \{\Phi_i,Q_j\} = \delta_{ij}.
\end{equation}

One may find it peculiar that the ``topological loops" $q_{l_6}$ and $q_{l_7}$ are nondynamical in the sense that they are not independent of loops on the surface of the torus. This is no accident. 
Indeed, for circuits on the fully connected graph on $V$ vertices, $K_V$, there are $\binom{V}{2}$ edges, and there must be at least $2g$ loops made of only capacitors or only inductors. 
To see how this is true, observe that the number of independent loops in a circuit $G$ with $E$ edges and $V$ vertices is at least
\begin{equation}
   \text{number of loops} \ge  n(E,V) = \text{max}\left(E - V + 1,0\right).
\end{equation}
We emphasize that this bound on $n(E,V)$ holds even for graphs that are not connected.  
Now, suppose the number of capacitors in a circuit is $N_C$, while the number of inductors is $N_I$. 
There exists a pair of graphs $G_C$ and $G_I$ that consist of all vertices from $G$ and only the capacitive edges or inductive edges respectively. 
The number of inductive (capacitive) loops in $G_I$ ($G_C$) is then at least $n(N_I,V)$ ($n(N_C,V)$). Since $G_C$ and $G_I$ are subgraphs of $G$, 
it follows that the number of homogenous loops in $G$ is at least $n(N_I,V) + n(N_C,V)$. Moreover, it must be the case that $N_I + N_C = E$. 
Thus, we have that the number of homogenous loops in $G$ is bounded below by 
\begin{equation}
N_l = n(N_I,V) + n(N_C,V) = n(E - N_C, V) +n(N_C, V).
\end{equation}
Now, for fully connected graphs $K_V$, $E = \binom{V}{2}$, so a circuit on $K_V$ with $N_C$ capacitors has 
\begin{equation}
    N_l = n\left(\binom{V}{2} - N_C, V\right) +n(N_C, V) \geq 2 g
\end{equation}
since 
\begin{equation}
    \min_x \left[ \max\left(\binom{V}{2 } -x - V + 1,0\right) + \max\left(x - V + 1,0\right)\right] \geq 2 g.
\end{equation}
In a few words, a fully connected graph on $V$ vertices can always be drawn on a $g$--holed torus in such a way that all $2 g$ topological loops are null vectors of $M$. It is, in principle, always possible to synthesize some planar circuit with the same dynamics as a given circuit on $K_5$. 

While it is generally impossible to construct a dual circuit for a nonplanar circuit, it is nonetheless sometimes possible to do something similar. In particular, it is sometimes possible to reduce a nonplanar circuit to an effectively planar circuit by using constraints that follow from Kirchoff's rules (e.g. by adding inductors or capacitors in series or parallel or by using the delta--wye transform \cite{otten_planarization}) . 
One way to accomplish this goal is to arrange circuit elements such that the edges lying on a given topological cycle contain either all capacitors or all inductors (so that the relevant loop degree of freedom is nondynamical). 
Then, one can produce a simplified equivalent circuit with  fewer edges (by using Kirchoff's rules) than the original. 
If the simplified circuit is then planar, and has a well--defined dual. In some sense, the resulting dual serves as what would be the dual of the nonplanar circuit. 
Whether or not this planarization may be carried out apparently depends on whether or not there exist at least $2g$ homogenous loops in a circuit. 
We mention in passing that no similar result holds for circuits on, say $K_{3,3}$. From this we suspect that the most interesting nonplanar circuits are likely to be those that are most sparsely connected. 

\section{Conclusions}
\begin{figure}
    \centering
    \begin{circuitikz}[scale=1.8]
        \draw[thick] (0,0) to[short,i^=$ $] (4,0) to[short,i^=$ $] (4,4); 
        \draw[thick] (0,0) to[short,i_=$ $] (0,4) to[short,i_=$ $] (4,4); 
        \draw[thick,red] (2,0) to[short,l_=$ $,color=red] (2,1) to[inductor,color=red] (2,2) to[inductor,color=red] (2,3) to[capacitor,color=red] (2,4); 
        \draw[thick,red] (2.1,0) to[short,i_=$\dot q$,color=red] (2.1,1) to (2.1,3) to (2.1,4);
        \draw[thick,blue] (0,2) to[short,l_=$ $,color=blue] (1,2) to[inductor,color=blue] (2,2) to[capacitor,color=blue] (3,2) to[inductor,color=blue] (4,2);
        \draw[thick] (1,2) to (1,1) to[capacitor] (2,1)to (3,1) to[inductor] (3,2) to (3,3) to[capacitor] (2,3) to (1,3) to[capacitor] (1,2);
        \draw[densely dashed] (1.5,1.5) to (2,1.5) to[short,i_=$\phi_{\text{dual}}$] (2.8,1.5);
        \draw[densely dashed] (1.5,2.5) to (2,2.5) to[short,i_=$\phi_{\text{dual}}$] (2.8,2.5);
        \draw[densely dashed] (1.5,3.5) to (2,3.5) to[short,i_=$\phi_{\text{dual}}$] (2.8,3.5);
    \end{circuitikz}
    \caption{A heuristic drawing on an identified polygon representation of a torus. If a charge or flux about some ``topological" loop is to be dynamical, then it's dual must also be topological. A ``topological current" is drawn and labeled in red. The naive dual to the current $\dot q$ would be a voltage drop across the corresponding cut. In the figure above, the nonlocal ``topological" voltage resulting from the naive duality map would participate in all three dual edges with the label $\phi_{\text{dual}}$.}
    \label{fig:enter-label}
\end{figure}
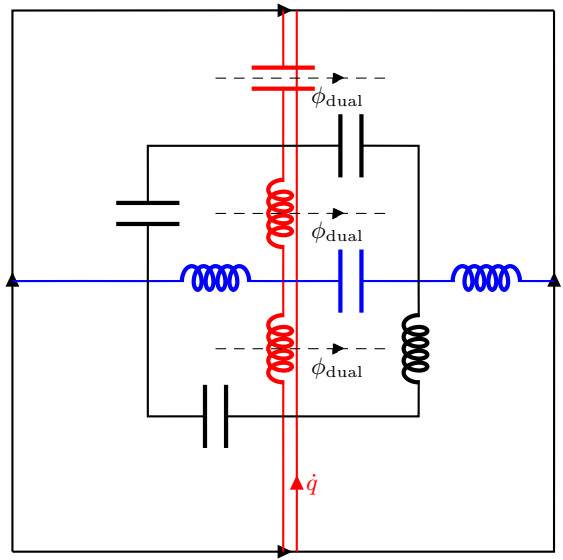

In this paper, we have derived a new Lagrangian description of circuits that treats fluxes and charges on an equal footing, extending our previous work \cite{osborne2023symplectic} towards a more general theory of circuit quantization. 
Existing methods can be shown to be equivalent to our formalism.
A key feature of this ``flux-charge symmetric" formulation of circuit mechanics is that circuit duality becomes a simple relabeling transformation, for planar circuits. 

For non-planar circuits, it is not always obvious if a dual circuit is physical. An interesting future research direction is to understand whether it is possible to add \emph{additional} circuit elements such that the dual of a general non-planar circuit is well-defined.  


\section*{Acknowledgements}
We thank Andr\'as Gyenis for helpful discussions.  This work was supported by a Research Fellowship from the Alfred P. Sloan Foundation under Grant FG-2020-13795 (AL) and by the U.S. Air Force Office of Scientific Research under Grant FA9550-21-1-0195 (AO, AL).

\begin{appendix}


\section{Graph theory}\label{app:graphthry}
Where possible, the demonstrations in this appendix will follow \cite{osborne2023symplectic} as closely as possible. 
\begin{defn}[Graph]
    Take a set $\mathcal V$ and let $\mathcal E$ be the set of ordered pairs of elements of $\mathcal V$. We say that $G = (\mathcal V, \mathcal E)$ is a (directed) \textbf{graph}.
\end{defn}
In all of what follows, we will assume that graphs and structures built upon graphs are connected. This assumption will not henceforth be stated, but we remark that it is not  required by every result we prove. 
\begin{defn}[Embedding]\label{defn:embedding}
Let $G = (\mathcal V, \mathcal E)$ be a graph. Choose a two dimensional orientable surface $S$, and associate one (distinct) point in $S$ for every element of $\mathcal V$: crudely we can think of $\mathcal{V}\subset S$. Every edge $e = (u,v) \in \mathcal E$ corresponds to a smooth map $e:[0,1]\rightarrow S$ such that $e(0)=u$ and $e(1)=v$. Assuming that for $x,x^\prime\in (0,1)$, $e(x)\ne e(x^\prime)$ (i.e. no two lines ever cross), $S$ is partitioned into a set of surfaces. 
If such a drawing exists, we call it an \textbf{embedding} of $G$ on $S$, or simply an embedding of $G$. 
\end{defn}
A single graph can have many embeddings. The combinatorial information in $G$ is independent of the manner in which it is embedded. 
If a graph $G$ can be embedded upon a surface $S$ at all, it is always possible to draw some embedding of $G$ on $S$ such that the surface $S$ is partitioned into regions isomorphic to the unit disk, though we remark that Definition \ref{defn:embedding} does not require this property. Nonetheless, these ``nice" embeddings will be the subject of our discussion. 
\begin{defn}[Face]
    In an embedding of $G$ on $S$, each ``patch" topologically equivalent to the unit disk is a \textbf{face}.
    For two faces $f_1$ and $f_2$  and edge $e$ of $G$ embedded on $S$, we write $f_1 \xleftrightarrow{e} f_2$ to mean that $e$ is an edge on the common boundary of the disk corresponding to $f_1$ and the disk corresponding to $f_2$.
\end{defn}
A graph $G$ is a purely combinatorial object.  We emphasize that while there is a universal notion of topology in the one-dimensional setting, the topology of the embedding is non-universal.  Perhaps surprisingly at first glance, even the subsets of edges in $\mathcal{E}$ which correspond to the boundaries of faces in an embedding is not unique, and can depend on the choice of embedding. As we will see in later discussion, these ambiguities certainly will  have no physical consequences in circuit quantization, although they can lead to different looking Lagrangians or Hamiltonians at first glance!    
\begin{defn}[Loop] \label{defn:loop}
Let $G = (\mathcal V, \mathcal E)$ be a graph. For the purposes of this definition, we regard $e = (u,v)$ and $e' = (v,u)$ as equivalent. 
A \textbf{loop} $\gamma$ is a subset $\gamma \subset \mathcal E$, so that (properly orienting each edge)
\begin{equation}
\gamma = \{(v_0, v_1), (v_1,v_2), \dots, (v_{n-1}, v_n), (v_n, v_0)\}.
\end{equation}
We say that $\gamma$ is of length $n$ if $|\gamma| = n$. 
\end{defn}
Every face's boundary is a loop, so there is a natural injection from the set of faces to the set of loops.
\begin{defn}[Linearly independent loops]
Let $G = (\mathcal V, \mathcal E)$ be a graph. 
Let $\Gamma = \{\gamma_1, \gamma_2, \dots, \gamma_n\}$ be a set of $n$ loops in $G$.
We say that the set $\Gamma$ is  linearly dependent if there exists $m\in\lbrace 1,\ldots, n\rbrace$ and  set $\iota = \{i_1,i_2,\dots,i_p\}\subset \{1,2,3,\dots, n\}\setminus\lbrace m\rbrace$,  such that 
\begin{equation}
    \gamma_m = \gamma_{i_1} \Delta \gamma_{i_2} \Delta \dots \Delta \gamma_{i_p}
\end{equation}
where 
\begin{equation}
    A \Delta  B = (A \cup B ) \setminus (A \cap B).
\end{equation}
Otherwise, we say that the loops in $\Gamma$ are \textbf{linearly independent}.
We say that the rank of a set of loops $\Gamma$ is the number of linearly independent loops in $\Gamma$. 
\end{defn}
\begin{defn}[Cut]
Let $G = (\mathcal V, \mathcal E)$ be a graph. 
    Partition $\mathcal V$ into $m$ sets $\mathcal V_i$ with $i = 1,2,3\dots,m$ so that 
    \begin{equation}
        \mathcal V = \bigcup_{i=1}^m \mathcal V_i
    \end{equation}
    and 
    \begin{equation}
        \emptyset = \mathcal V_i \cap \mathcal V_j. 
    \end{equation}
    for $i \neq j$. 
    Such a partition $(\mathcal V_1, \mathcal V_2,\mathcal V_3, \dots,\mathcal V_m)$ of $\mathcal V$ is called a \textbf{cut} of $G$.  
\end{defn}
\begin{defn}[Edge--cut]
    Let $G = (\mathcal V,\mathcal E)$ be a graph. Let $(\mathcal V_1, \mathcal V_2 , \dots ,\mathcal V_m)$ be a cut of $G$. 
    Define the set $\mathcal K \subset \mathcal E$ such that 
    \begin{equation}
       \mathcal K = \{(u,v) \in \mathcal E \text{ s.t. }  \exists i \text{ s.t. } \{u,v\} \subset \mathcal V_i\}.
    \end{equation}
    The set $\mathcal E \setminus \mathcal K$ is called the \textbf{edge--cut} of $G$ induced by $(\mathcal V_1, \mathcal V_2, \dots, \mathcal V_m)$.  
\end{defn}

The notions of a cut and an edge--cut are formally equivalent and every result of this manuscript may be stated and proven favoring one notion over the other. Nonetheless, it is useful to posses an understanding of both conventions since some ideas are easier to understand in terms of cuts and others are easier to understand in terms of edge--cuts.
\begin{defn}[Genus]
  Let $G = (\mathcal V, \mathcal E)$ be a graph. We say that $G$ has genus $g$ if there exists no embedding of $G$ on any surface $S$ of topological genus $g' < g$.  This is not the same definition as the usual graph genus from graph theory.
\end{defn}
\begin{defn}[Embedded graph]
Let $G$ be a graph, and $S$ an orientable surface. Suppose that there exists some embedding of $G$ on $S$, i.e. on a sphere with $g$ holes. Under this embedding, denote the set of faces of $G$ embedded in $S$ as $\mathcal F$. 
We write $G_S = (\mathcal V, \mathcal E, \mathcal F)$ to be the \textbf{embedded graph} $G$ on $S$. 
\end{defn}

In much of what follows, we will suppress some of the above defined terminology when either the embedding of $G$ on $S$ is of no consequence, or when the embedding of $G$ on $S$ is clear from context. In such cases, we will write $G = (\mathcal V,\mathcal E,\mathcal F)$ and we will refer to $G$ simply as a graph. 
\begin{thm}\label{thm:facefull}
    Let $G = (\mathcal V, \mathcal E, \mathcal F)$ be an embedded graph. Denote the loop on the boundary of the face $f_i$ as $\gamma_i$.  If $G$ is connected, the set 
    \begin{equation}
        \Gamma = \{\gamma_1, \gamma_2, \dots, \gamma_{|\mathcal F|}\}
    \end{equation}
    is linearly dependent
    and the set $\Gamma \setminus \{\gamma_1\}$ is linearly independent.
\end{thm}
\begin{proof}
    Every edge in $G$ appears on the boundary of precisely two faces, so $\gamma_1 = \gamma_2 \Delta \gamma_3 \Delta \dots \Delta \gamma_{|\mathcal F|}$.  For any $\gamma_j$ obeying $\gamma_1\cap\gamma_j \ne \emptyset$, $\gamma_j$ is linearly independent of the rest, since every edge appears in no other face.  Removing such $\gamma_j$, we can inductively deduce the linear independence of the remaining $\gamma$s, since the graph is connected.
\end{proof}

The following two important, and classic, results, are stated without proof:

\begin{thm}[Euler's formula]\label{thm:euler}
    Let $G = (\mathcal V, \mathcal E)$, and suppose that $G$ has genus $g$. Let $S$ be a compact, orientable discretization of a Riemann surface of genus $g$. For any embedding of $G$ on $S$, we have 
    \begin{equation}
        |\mathcal V| - |\mathcal E| + |\mathcal F| = 2 - 2g.
    \end{equation}
\end{thm}
\begin{cor}\label{cor:independence}
    Let $G = (\mathcal V, \mathcal E)$ be a graph and let $\Gamma$ be a linearly independent set of loops in $G$. 
    Then 
    \begin{equation}
        |\Gamma| \leq g_G = |\mathcal E| - |\mathcal V| + 1.
    \end{equation}
\end{cor}
In effect, Cor.~\ref{cor:independence} serves to indicate that, for nonplanar graphs, it is possible to construct a linearly independent set of loops that is greater in size than the set of faces of a graph. 
As we will see, this fact is crucial to our discussion of circuit dualities for nonplanar graphs.

\begin{thm}
    Let $G = (\mathcal V, \mathcal E, \mathcal F)$ be an embedded, nonplanar graph with genus $g>0$.
    There exist $2g$ loops that are independent of all of the loops bounding faces $f$ in $\mathcal F$. 
\end{thm}
\begin{proof}
    By Theorem \ref{thm:euler}, 
    \begin{equation}
        |\mathcal F| + 2 g = g_G + 1.
    \end{equation}
    As we have seen in Theorem \ref{thm:facefull}, the set of loops bounding faces in $\mathcal F$ is $|\mathcal F| - 1$, and Corollary~\ref{cor:independence} tells us that it is possible to construct a set of loops of rank $g_G$.  
\end{proof}

\begin{defn}[Loop set]
    Let $G = (\mathcal V, \mathcal E, \mathcal F) $ be an embedded graph of genus $g$. Choose $2g$ loops independent of all loops at the boundary of some face in $\mathcal F$, say $l_1, \dots , l_{2 g}$. Denote the loop at the boundary of $f_i$ in $\mathcal F$ as 
    $l_{2g + i}$. 
    The set 
    \begin{equation}
        \mathcal L = \{l_1,l_2, \dots,l_{2g+|\mathcal F|}  \}
    \end{equation}
    is called the \textbf{loop set} of $G$. 
\end{defn}

\begin{defn}[Extended embedded graph]
    Let $G = (\mathcal V, \mathcal E, \mathcal F)$ be an embedded graph, and let $\mathcal L$ be the loop set of $G$.
    The object $G' = (\mathcal V, \mathcal E, \mathcal L)$ is called an \textbf{extended embedded graph}.
\end{defn}

\begin{defn}[Planar graph]
A graph $G$ is \textbf{planar} if it has genus $g=0$.
\end{defn}
We remark that, for planar graphs, $\mathcal F = \mathcal L$.

Planar graph theory is well--studied. In the circuit literature, many formal results are limited to planar graphs in consideration because such circuits are both simpler to draw and analyze, let alone build in experiment. Note that according to Definition \ref{defn:embedding}, a plane is not a suitable surface on which to embed a graph because the ``external face" cannot be isomorphic to a unit disk. This problem is resolved by considering the embedding of a planar graph on a sphere or equivalently by identifying the points at infinity to be equivalent. Of course in practice, we will draw planar graphs in the plane, since the plane is equivalent to the sphere if we identify all points at spatial infinity with the same point.
So henceforth, we will consider planar graphs to be embedded on the unit sphere.

Nonplanar graphs have genus $g>0$, by definition. 
They do have much in common with planar graphs, once we find the right perspective.  To talk about graph duality, it is necessary to specify surfaces on which we consider embedding the graph. Generally, for a graph of genus $g$, we elect 
to embed upon a ``sphere with $g$ handles", or equivalently a  ``torus with $g$ holes". 

\begin{defn}[Chains]
    Let $X$ be a finite set. For every $x \in X$, define a real vector $|x\rangle$ and define 
    \begin{equation}
        \mathcal D(X) = \mathrm{span}\left(\{|x\rangle: x \in X\}\right).
    \end{equation}
    We say that $\mathcal D(X)$ is the \textbf{set of chains} over $X$, and we say that $|x\rangle$ is a chain.
\end{defn}
\begin{defn}[Incidence matrix]
    Let $G = (\mathcal V, \mathcal E)$ be a directed graph. Define the linear map $A: \mathcal D(\mathcal V) \rightarrow \mathcal D(\mathcal E)$ so that
    \begin{equation}
        \langle e| A|v\rangle = \begin{cases}
            1 & e \text{ is incident upon } v \\ 
            -1 & e \text{ leaves } v\\
            0 & \text{ otherwise }
        \end{cases}.
    \end{equation}
    We say that $A$ is the \textbf{incidence matrix } of $G$ and we often write $\langle e| A|v\rangle = A_{ev}$.
\end{defn}
\begin{defn}[Orientation matrix]
    Let $G = (\mathcal V, \mathcal E, \mathcal L)$ be an extended embedded graph. Orient all faces of $G$ alike\footnote{Our convention is that loops bounding faces ought to be oriented such that the right hand rule points out of the surface.}. Choose some orientation for the $2g$ loops in $\mathcal L$ that do not bound a face of $G$. 
    Define the linear map $B: \mathcal D(\mathcal E) \rightarrow \mathcal D(\mathcal L)$  so that 
    \begin{equation}
        \langle l | B|e\rangle = \begin{cases}
           1 & e \text{ borders } l  \text{ and } e \text{ is oriented with } l \\
           -1 & e \text{ borders } l  \text{ and } e \text{ is oriented against } l \\ 
           0 & \text{ otherwise}
        \end{cases}.
    \end{equation}
    We say that $B$ is the \textbf{orientation matrix} of $G$.  
\end{defn}
An abstract graph $G$ has a unique incidence matrix $A$, but not necessarily a unique orientation matrix $B$. On the other hand, an extended embedded graph on a surface $S$, $G_S$ has both a unique $A$ and a unique $B$. 
\begin{thm}\label{thm:rank}
  Let $G$ be an extended embedded graph with incidence matrix $A$ and orientation matrix $B$. The rank of $A$ is $|\mathcal V| - 1$ and the rank of $B$ is  $|\mathcal L| - 1$.
\end{thm}
\begin{proof}
  Suppose
  \begin{equation}
    \sum_{v} A_{ev} c_v = 0.
  \end{equation}
  with $c_v$ in $\mathbb R^{|\mathcal V|} \setminus \{0\}$.
  Then, there exists at least  $v$ such that $c_v$ is nonzero, and 
  if $A_{ev}$ is nonzero then there exists $u$ so that  $e = (u,v)$ or  $e = (v,u)$. In either case, 
  $A_{ev} = - A_{eu}$ and thus $c_u = c_v$. Continue in this way to discover that $c_v = 1$ for all $v \in \mathcal V$.
  The proof that the only left null vector of  $B$ is given by $c_f = 1$ for all $f \in \mathcal F$ is exactly analogous. We remark that $\mathcal F \neq \mathcal L$ in general; hence the left null vector of $B$ corresponds to a sum over loops bounding faces only (see Theorem \ref{thm:facefull}). 
\end{proof}
\begin{thm}\label{thm:exact}
    If $G$ be an extended embedded graph with incidence matrix $A$ and orientation matrix $B$,
    \begin{equation}
        \sum_{e \in \mathcal E} B_{le} A_{ev} = 0.
    \end{equation}
\end{thm}
\begin{proof}
  This fact is an elementary result in the theory of CW--complices\footnote{Explicitly, this is a discrete version of the statement that exterior derivatives are nilpotent.}. Nonetheless, we provide an explicit demonstration. 
  Choose a loop $l$ and a vertex $v$. If $v$ is not connected to any edges in $l$, the result is trivial so suppose that there exist edges $e$ and $e'$ in\footnote{Each vertex along a loop must be hit an even number of times, or the loop would not be closed.  If a vertex is hit $2m$ times, then there will always exist $m$ pairs of $e$ and $e^\prime$ for which this argument holds.} $l$ and vertices $u_1$ and $u_2$ so that one of the following possibilities holds:
  \begin{enumerate}
      \item $e = (u_1, v)$ and $e' = (v,u_2)$ 
      \item $e = (v, u_1)$ and $e' = (v,u_2)$ 
      \item $e = (u_1, v)$ and $e' = (u_2,v)$
      \item $e = (v,u_1)$ and $e' = (u_2,v)$.
  \end{enumerate}
  Now, in cases 1. and 4., $B_{le} = B_{le'}$ and $A_{ev} = - A_{e'v}$. In cases 2. and 3., the opposite is true. In any case,
  \begin{equation}
      B_{le}A_{ev} + B_{le'}A_{e'v} = 0
  \end{equation}
  This concludes the proof. 
\end{proof}
\begin{thm}\label{thm:planargood}
    Let $G$ be an extended embedded graph with incidence matrix $A$ and orientation matrix $B$. 
    Then,
    \begin{equation}\label{eq:kerbima}
      \mathrm{Ker}(B) = \mathrm{Im}(A).
    \end{equation}
\end{thm}
\begin{proof}
  From Theorem \ref{thm:exact}, it is clear that 
  \begin{equation}
    \text{Im}(A) \subseteq \text{Ker}(B).
  \end{equation}
  For the other direction, note that 
  \begin{subequations}\begin{align}
    \text{Rank}(A) &= |\mathcal V| - 1, \\
     \text{Rank}(B)  &= |\mathcal L| - 1.
  \end{align}\end{subequations}
  Since extended embedded graphs satisfy
  \begin{equation}\label{eqn:eiler}
    |\mathcal V| - |\mathcal E| + |\mathcal L| = 2,
  \end{equation}
  we see that \begin{equation}
      \dim (\mathrm{Ker}(B)) = |\mathcal{E}| - \text{Rank}(B) = |\mathcal{E}|-|\mathcal{L}|+1 = |\mathcal{V}|-1 = \text{Rank}(A) = \dim(\mathrm{Im}(A)). \label{eq:kerbima2}
  \end{equation}
  It follows from (\ref{eq:kerbima}) and (\ref{eq:kerbima2}) that  $\text{Ker}(B)$ is spanned by elements of  $\text{Im}(A)$. 
\end{proof}
Another way to state Theorem \ref{thm:planargood} is that for any vector $\gamma$ in $\mathbb R^{|\mathcal E|}$ satisfying
\begin{equation}
  \sum_{e} B_{le} \gamma_e = 0,
\end{equation}
there exists some other vector $\delta$ in  $\mathbb R^{|\mathcal V|}$ such that 
 \begin{equation}
   \gamma_e = \sum_{v} A_{ev} \delta_v.
\end{equation}
Likewise any left null vector $\gamma$ of $A$ is of the form $\gamma = \delta^T B$.

Informally, $A$ maps vertices (or integer linear combinations of vertices) to edge--cuts of a graph $G$, and such edge--cuts are precisely the null vectors of $B$.
More precisely, if $(\mathcal V_1, \mathcal V_2)$ is a cut of $G$, then 
$\langle e| A \sum_{v \in \mathcal V_1} |v\rangle$ is nonzero if and only if $e$ is in the edge--cut induced by $(\mathcal V_1, \mathcal V_2)$. The sign may either be positive or negative and is determined by the relative orientation of edges leaving or entering $\mathcal V_1$. 
Likewise, the matrix $B$ maps cycles of $G$ to linear combinations of faces of $G$. 
As a remark, we note that, by cutting every edge in a graph, $\mathcal V$ is partitioned into singlet sets, and since the graphs of interest are also circuits, there is some sense in which the cycle consisting of all edges \emph{is}  indeed a loop. 
However, the vector $\sum_{e \in \mathcal E} |e\rangle$ is not in the range of $A$.

\section{Formal approach to symmetric quantization}\label{app:symmquant}
In this appendix, we discuss more formally the flux-charge symmetric theory of circuit quantization.
\subsection{Properties of circuits and the connection matrix}
In this section, we describe the full formalism for symmetric quantization on arbitrary graphs. In what follows, we consider a circuit on a planar graph with a fixed  but arbitrary embedding. In analogy to graphs as combinatorial objects, we regard circuits as graphs with ``colored" edges. That is to say that we demand that edges contain precisely a single circuit element and that circuit elements are either inductive or capacitive.

\begin{defn}[Circuit]
    Let $G = (\mathcal V, \mathcal E, \mathcal L)$ be an extended embedded graph with incidence matrix $A$ and orientation matrix $B$. 
    Partition the set $\mathcal E$ into the sets $\mathcal C$ and $\mathcal I$ such that edges housing inductors (capacitors) go into the set $\mathcal I$  ($\mathcal C$). 
    The tuple $(\mathcal V, \mathcal C, \mathcal I, \mathcal L, A , B)$ is called a \textbf{circuit}.
\end{defn}
\begin{defn}[Connection matrix]
  \introem
  Define the matrix $M: \mathcal D (\mathcal V) \rightarrow \mathcal D(\mathcal L)$  so that 
  \begin{equation}
    M_{lv}:=\langle l | M | v\rangle = \frac{1}{2}\sum_{e \in \mathcal C} B_{le}A_{ev} - \frac{1}{2}\sum_{e \in \mathcal I} B_{le}A_{ev}.
  \end{equation}
  We say that $M$ is the \textbf{connection matrix} of $G$. 
\end{defn}
Unlike $A$ and $B$, $M$ has no intuitive interpretation which can be easily read off of a circuit (at least that we have found).  Still, in practice, it is straightforward to simply calculate it.
Since $\mathcal C \cup \mathcal I = \mathcal E$, it follows that 
\begin{equation}\label{eqn:Msame}
   M = \sum_{e \in \mathcal C} B_{le}A_{ev} = - \sum_{e \in \mathcal I} B_{le}A_{ev}.
\end{equation}
It will often useful to rewrite $M$ using (\ref{eqn:Msame}).
Furthermore, for planar circuits, 
\begin{equation}
  \sum_{l \in \mathcal L} M_{lv} = \sum_{v \in \mathcal V} M_{lv} = 0
\end{equation}
since $\sum_{v} A_{ev} = 0$ and  $\sum_{l} B_{le} = 0$. 
For nonplanar circuits, the set of faces $\mathcal{F} \subset \mathcal L$ has the property that
\begin{equation}
    \sum_{l \in \mathcal{F}} B_{le} = 0.
\end{equation}

Our first goal is to prove the following result: 
\begin{cor}
    \introem 
    We say that a loop is homogenous if the edges are all capacitors or all inductors. Similarly a cut is called homogenous if its induced edge--cut consists only of inductors or only of capacitors.
    The following conditions hold.
    \begin{itemize}
        \item $M|\varphi \rangle = 0$ if and only if $(\mathcal V_1, \mathcal V_2)$ is a homogeneous cut of $G$ and $|\varphi \rangle = \sum_{v\in \mathcal V_1} |v\rangle$ 
        \item $\langle \psi | M = 0$ if and only if $\gamma$ is a homogeneous loop and $\langle \psi | = \sum_{e \in \gamma} \langle e|$.
    \end{itemize}
\end{cor}
This result is the corollary of a more abstract mathematical result:
\begin{thm}\label{thm:nullvbetter}
Let $V_1$, $V_2$, and $V_3$ be vector spaces and let $A: V_1 \rightarrow V_2$ and $B: V_2\rightarrow V_3$ be linear maps satisfying\footnote{In other words, $B|\varphi\rangle = 0$ for $|\varphi\rangle \in V_2$ if and only if for some $|\varphi^\prime\rangle \in V_1$, $|\varphi\rangle = A|\varphi^\prime\rangle$.} 
\begin{equation}\label{eqn:hypo}
    \ker(B) = \mathrm{im}(A)
\end{equation}
The matrix $M = B P A : V_1 \rightarrow V_3$ has the following properties:
\begin{enumerate}
    \item $M |\varphi \rangle = 0 $ if and only if there exist vectors $|+\rangle$ and $|-\rangle$ such that $|\varphi\rangle = |+\rangle + |-\rangle$ and $P A | \pm\rangle = \pm A |\pm \rangle$; moreover, $|\pm\rangle$ are separately also right null vectors of $M$.
    \item $\langle \psi |M = 0 $ if and only if there exist vectors $|+\rangle$ and $|-\rangle$ such that $|\psi\rangle = |+\rangle + |-\rangle$ and $\langle \pm | B P = \pm\langle \pm| B $; moreover, $\langle \pm|$ are separately also left null vectors of $M$.
\end{enumerate}

\end{thm}
\begin{proof}
We prove point 1 of the list, as point 2 is proven analogously. Suppose 
\begin{equation}
    M |\varphi \rangle = B P A |\varphi \rangle = 0.
\end{equation}
Define the projectors
\begin{equation}
    K_{\pm} = \frac{1}{2}(\mathbb{I} \pm P)
\end{equation}
which have the property that $P K_{\pm} = \pm K_{\pm}$,  $K_{\pm}^2 = K_{\pm}$, and $K_{+} + K_{-} = \mathbb{I}$.  
It follows that 
\begin{equation}
   0 = BPA|\varphi\rangle = BP (K_+ + K_-) A |\varphi \rangle = B (K_+ - K_-) A|\varphi\rangle.
\end{equation}
Of course, we also know that \begin{equation}
    0 = BA|\varphi\rangle = B(K_+ + K_-)A|\varphi\rangle,
\end{equation}
meaning that
\begin{equation}
    B K_+ A|\varphi \rangle = B K_- A |\varphi \rangle  = 0.
 \end{equation}
 By (\ref{eqn:hypo}) we conclude that there exist $|\pm\rangle$ for which \begin{equation}
     A|\pm\rangle = K_\pm A|\varphi\rangle. \label{eq:Apm}
 \end{equation}
 Evidently,
 \begin{equation}
    A(|+\rangle + |-\rangle)   = (K_+ + K_-)A |\varphi \rangle = A|\varphi \rangle.
 \end{equation}
 Since any right null vector of $A$, $|n\rangle$, satisfies 
 \begin{equation}
     P  A |n\rangle = \pm  A|n\rangle = 0, 
 \end{equation}
 we see that $|+\rangle$ and $|-\rangle$ are only determined up to the addition of right null vectors of $A$, should any exist. 
 Therefore, $A(|+\rangle + |-\rangle ) = (K_+ + K_-)A|\varphi\rangle = A|\varphi\rangle$, which implies that  $|\varphi\rangle = |+\rangle + |-\rangle$ (for appropriate definitions of $|+\rangle $ and $|-\rangle$ corresponding to the freedom to add right null vectors of $A$ to either). 
 Left multiply (\ref{eq:Apm}) by $P$ to conclude that $PA|\pm\rangle = \pm A|\pm\rangle$.
 Since $PA|\pm\rangle$ is proportional to $A|\pm\rangle$ and $BA=0$, it follows that $M|\pm\rangle = 0$, which proves the desired statements.

\end{proof}

A straightforward, but useful, consequence of this result is:
\begin{cor}\label{cor:bnullv}
     Adopt the definitions made in the proof of Theorem \ref{thm:nullvbetter}. Consider the matrix 
     \begin{equation}
         W = B K_+. 
     \end{equation}
     If $|\varphi\rangle$ satisfies
     \begin{equation}
         W| \varphi \rangle = 0 ,
     \end{equation}
     then there exist vectors $|\theta\rangle$ and $|\psi\rangle$ such that 
     \begin{equation}
         |\varphi\rangle = A|\theta\rangle + K_- |\psi\rangle
     \end{equation}
     with $K_+ A | \theta\rangle = A |\theta \rangle$.
 \end{cor}

As a passing remark, the hypothesis of Theorem \ref{thm:nullvbetter} is satisfied by the vector spaces and boundary maps in any short exact sequence of vector spaces together with some partition of the intermediate vector space. For our purposes, Theorem \ref{thm:nullvbetter} serves to enumerate all of the null vectors of $M$. To see how Theorem \ref{thm:nullvbetter} applies to $M$, define $P:\mathcal D(\mathcal E)\rightarrow \mathcal D(\mathcal E)$ such that 
\begin{equation}
    P_{ee'}  = (-2 \mathbb{I}[e \in \mathcal I] + 1) \delta_{ee'}
\end{equation}
where $\mathbb I$ is an indicator function that vanishes if its argument is untrue and is otherwise equal to one. Then, 
\begin{equation}
    M = \frac{1}{2} B P A.
\end{equation}
To make more clear our enumeration of null vectors of $M$, we make the following definitions:  
\begin{defn}[Homogeneous cut]
    \introem
    Let $C = (\mathcal V_1, \mathcal V_2)$ be a cut of $G$. $C$ is a homogeneous cut if the set 
    \begin{equation}
        \mathfrak C = \{e \subset \mathcal E \,\,:\,\, \exists v_1\in \mathcal V_1 \text{ and } v_2 \in \mathcal V_2 \,\, \text{s.t. } e \in \{(v_1,v_2), (v_2,v_1)\}\}
    \end{equation}
    has the property that 
    \begin{equation}
        \mathfrak C \subset \mathcal C
    \end{equation}
    or 
    \begin{equation}
        \mathfrak C \subset \mathcal I.
    \end{equation}
    In the former case we say that $\mathfrak C$ is capacitive and in the latter case we say that $\mathfrak C $ is inductive. 
\end{defn}
\begin{defn}[Homogeneous loop]
    Let $\gamma$ be a loop in the sense of Definition \ref{defn:loop}. If $\gamma \subset \mathcal C$ or $\gamma \subset \mathcal I$, we say that $\gamma$ is a \textbf{homogeneous loop}.
\end{defn}

  \begin{cor}\label{cor:gtry}
    \introem
    Let $\Delta_I$  ($\Delta_C$) be the set of homogeneous inductive (capacitive) loops of $G$, and let $\Gamma_I$ ($\Gamma_C$) be the set of homogenous inductive (capacitive) loops in $G$. Then, 
    \begin{equation}
      |\mathcal F| - |\Delta_I| - |\Delta_C| = |\mathcal V| - |\Gamma_I| - |\Gamma_C| 
    \end{equation}
  
  \end{cor}
  \begin{proof}
    This is an immediate consequence of Thm. \ref{thm:nullvbetter} together with the Rank--Nullity theorem. 
  \end{proof}

  Corollary \ref{cor:gtry} was a result nearly achieved in \cite{osborne2023symplectic}, but the relevant discussion relied upon the enumeration of null vectors \emph{and} a number of Noether currents. 
  One of the merits of this approach is that Noether currents in the formalism of \cite{osborne2023symplectic} are promoted to null vectors. 

     Thus, the number of degrees of freedom in a circuit is equal to the number of loops in a circuit which are neither purely inductive nor purely capacitive. 

  \subsection{Formal circuit Lagrangian}
  Here, we restrict our attention to planar graphs. These results can be made to hold for nonplanar graphs with some minor modifications, which will be made explicit in appendix \ref{app:duality}.
  \begin{defn}[Symmetric circuit Lagrangian]
    \introem
    Let $M$ be the connection matrix of $M$ and define $|\mathcal C \cup \mathcal I|$ functions labeled  $E_e$ which describe the energy of the circuit element on branch $e$. 
    The function 
    \begin{equation}\label{eqn:symlag}
      L = \sum_{l \in \mathcal L, v \in \mathcal V} q_l M_{fv} \dot \phi_v - \sum_{e \in \mathcal C} E_e\left(\sum_l q_l B_{le}\right) - \sum_{e \in \mathcal I} E_e\left(\sum_{v} A_{ev}\phi_v\right)
    \end{equation}
    is called the \textbf{symmetric circuit Lagrangian} for $G$, or ``the Lagrangian for $G$" for short.
  \end{defn}
  We see that, from (\ref{eqn:symlag}), $S = \int \mathrm d t L$ is symmetric in  $\mathcal C$ and $\mathcal V$. As we will discuss later, 
  this choice of variables leads to a very straightforward circuit duality transformation. 

    For the following result, we will need to rely upon the results of Section \ref{sec:bnf} as well as a number of definitions originally made in \cite{osborne2023symplectic}.
  \begin{thm}\label{thm:canquant}
   \introem
   Let $L$ be the symmetric circuit Lagrangian of $G$.
   Define $\Gamma_C$ ($\Gamma_I$) to be the set of capacitive (inductive) cuts of $G$, and let $\Delta_C$ ($\Delta_I$) to be the set of capacitive (inductive) cycles of $G$.
   It is always possible to define $|\mathcal V| - |\Gamma_I| - |\Gamma_C|-1$ variables $Q_i = \sum_l D_{il} q_l$ and $\Phi_i = \sum_{v} S_{iv} \phi_v$ so that 
   \begin{equation}
     \sum_{l,v} q_l M_{lv}\dot\phi_v = \sum_{i = 1}^{|\mathcal V| - |\Gamma_I| - |\Gamma_C|-1} Q_i \dot \Phi_i.
   \end{equation}
   All possible choices of $S$ and $D$ are related by canonical transformations. 
  \end{thm}
  \begin{proof}
    Define $q_e = \sum_{l} q_l B_{le}$ for  $e \in \mathcal C$, and define $\Omega_{ev} = A_{ev}$ for  $e \in \mathcal C$ and then apply Theorems 10 and 13 from \cite{} directly.  
  \end{proof}

  \section{Circuit duality}\label{app:duality}
  The contents of this appendix depend broadly on the results of Appendix \ref{app:symmquant} and serve to formalize the claims in Section \ref{sec:duality}.
  Duality as a map can be sensibly defined on Lagrangians, graphs, and structures from the theory of topological algebra. 
  While we take a minimal perspective here, the results of this section are very simply expressed as a property of chain--complex isomorphisms. 
  \begin{defn}[Hamiltonian duality transformation]
      Suppose $H$ is a Hamiltonian function of variables $\Phi_i$ and $Q_i$ for $i = 1,2,\dots,N$, equipped with Poisson brackets
      \begin{equation}
          \{\Phi_i,Q_j\} = \delta_{ij}.
      \end{equation}
      The transformation 
      \begin{equation}
          \begin{aligned}
              Q_i \rightarrow Q_i' = -\Phi_i \\
              \Phi_i \rightarrow \Phi_i' =  Q_i
          \end{aligned}
      \end{equation}
      is called a \textbf{Hamiltonian duality transformation}. We will write $H(Q',\Phi') = H^*$.
  \end{defn}
  Clearly, Hamiltonian duality transformations are canonical since  $\{Q_i', \Phi_j'\} = \{Q_i, \Phi_j\}$. Certainly, at the level of Hamiltonian mechanics, it is straightforward to take the dual of any Hamiltonian arising from a circuit Lagrangian in the spirit of the formalism of this work. 
  However, the challenge of constructing the circuit (or circuits) that produce $H^*$ is the subject of this appendix. 
  Moreover, it is not always possible to produce a (physically sensible) circuit that accomplishes this task -- at least using any known algorithm for constructing a dual circuit.
  \begin{defn}[Dual Circuit]\label{defn:dualityset}
      \introem
      Define
      \begin{equation}
          \begin{aligned}
              \mathcal V^* &= \mathcal L \\ 
              \mathcal I^* &= \mathcal C \\
              \mathcal C^* &= \mathcal I \\ 
              \mathcal L^* &= \mathcal V \\ 
              A^* &= B^{\mathrm T} \\ 
              B^* &= A^{\mathrm T}
          \end{aligned}
      \end{equation}
      and finally 
      \begin{equation}
          G^* = (\mathcal V^*, \mathcal C^*, \mathcal I^*, \mathcal L^*, A^*, B^*).
      \end{equation}
      We say that $G^*$ is the \textbf{dual circuit} of $G$. 
      For an element $v$ of $\mathcal V$, we write the corresponding element of $\mathcal L^*$ as  $v^*$.
      For an element $l$ of $\mathcal L$, we write the corresponding element of $\mathcal V^*$ as  $l^*$. 
  \end{defn}
  Our reason for using this terminology will become clear shortly. 
  Both a combinatorial object and a topological object are encoded in $G^*$. That is to say that the combinatorial properties of $G^*$ are encoded in the structure of $A^*$.
  While it is straightforward to recover the combinatorial structure of $G^*$ by looking at the matrix $A^*$, it is less obvious how one might recover a particular embedding of $G^*$ by using $B^*$. 
  Though we will not belabour this point presently,  we remark that a particular embedding of $G^*$ is recoverable by a gluing procedure where every element of $\mathcal F$ is represented by a patch isomorphic to the unit disk, and then patches are glued together by identifying segments on the boundary of different patches, in a way that is consistent with the content of $A^*$. This is always possible. For planar graphs, this procedure always accomplished by the following procedure:
\begin{defn}[Embedding of dual circuit]\label{def:duality}
\introem
Further suppose $G$ is embedded upon a sphere.
We construct the \textbf{embedded dual circuit} of $G$, $G^*$ as follows:
\begin{enumerate}
\item For every loop $l$ in $\mathcal L$, draw a vertex labeled $l^*$ (inside of the face whose boundary is $l$). 
\item For every pair of loops $l_0$ and $l_1$ in $\mathcal L$, draw an inductive (capacitive) edge between $l_0^*$ and $l_1^*$ for every capacitive (inductive) edge in both $l_0$ and $l_1$. For such an edge $e$ of $G$, label the corresponding edge in $G^*$ as $e^*$. 
\item For every edge $e^*$ in $G^*$, give $e^*$ an orientation so that when the surface is drawn (locally) on the plane, the cross product between $e$ and $e^*$ is always positive.
\end{enumerate}
\end{defn}
Definitions \ref{def:duality} and \ref{defn:dualityset} are equivalent for planar graphs. 
We emphasize that the construction of a dual graph is intrinsically dependent upon the embedding of $G$ chosen. We will make this point explicit with the next observation.
\begin{obs}
  Let $G$ be a planar graph. Choose two embeddings of $G$ on $S$, $G_S$ and $G_S'$. $G_S^*$ and $(G_{S'})^*$ need not be graph isomorphic.
  \end{obs}
\begin{proof}
We provide a proof by example in Figure \ref{fig:ulrich}(a)--(b). In order to view a circuit as a graph, one only needs to ignore all of the circuit elements in the circuit so that every branch becomes simply a graph theoretic edge. 
\end{proof}

\begin{obs}
  \introem
  Let $G^* = (\mathcal V^*, \mathcal C^* ,\mathcal I^*, \mathcal L^*, A^*,B^*)$ be the dual circuit of  $G$. 
  Define  the matrix $M^*: \mathcal V^* \rightarrow \mathcal L^*$ with matrix elements
  \begin{equation}
    M^*_{v^*,l^*} = \frac{1}{2}\sum_{e^* \in \mathcal C^*} B^*_{v^* e^*} A^*_{e^* l^*} - \frac{1}{2}\sum_{e^* \in \mathcal I^*} B^*_{v^*e^*}A^*_{e^* l^*}.
  \end{equation}
  Then, $M^* = - M^{\mathrm T}$. 
\end{obs}
\begin{proof}
  By relabeling, 
  \begin{equation}
    \langle v^* | M^*|l^*\rangle = \frac{1}{2}\sum_{e^* \in \mathcal C^*} B^*_{v^* e^*} A^*_{e^* l^*} - \frac{1}{2} B^*_{v^*e^*}A^*_{e^* l^*} = 
    \frac{1}{2}\sum_{e \in \mathcal I} A_{v e} B_{e l}  - \frac{1}{2} \sum_{e \in \mathcal C} A_{ve}B_{el}  =- \langle v | M^{\mathrm T} |l\rangle.
  \end{equation}
  Since $\langle v^* | M^* | l^* \rangle = \langle v| M^{\mathrm T} |l \rangle$, we say simply  $-M^{\mathrm T} = M^*$.
\end{proof}

\begin{thm}\label{thm:dualitylag}
  \introem
  Let $G^*$ be the dual circuit of  $G$. For each edge $e$ in $\mathcal I \cup \mathcal C$, suppose that the energy associated with edge $e$ is a function  $E_e: \mathbb R\rightarrow \mathbb R$.
  For each edge  $e^* $ in $\mathcal I^* \cup \mathcal C^*$, fix  $E_{e^*} : \mathbb R \rightarrow \mathbb R$ so that 
  \begin{equation}
    E_{e^*}(x) = E_e(x).
  \end{equation}
  The Lagrangian for $G$ is related to the Lagrangian for $G^*$ are related by a relabeling transformation.
\end{thm}
\begin{proof}
  The Lagrangian for $G$ is given by 
  \begin{equation}
    L = \sum_{l \in \mathcal L} \sum_{v \in \mathcal V} Q_l M_{l v} \dot \phi_v - \sum_{e \in \mathcal I} E_e\left(\sum_{v}A_{ev}\phi_v\right) - \sum_{e \in \mathcal C} E_e\left(\sum_l q_l B_{le}\right).
  \end{equation}
  On the other hand, the Lagrangian for $G^*$ is given by 
   \begin{equation}
     L^* = \sum_{v^* \in \mathcal L^*} \sum_{l^* \in \mathcal V^*} q_{v^*} M^*_{v^* l^*}\dot \phi_{l^*} - \sum_{e^* \in \mathcal I^*} E_e\left(\sum_{l^*\in\mathcal V^*}A^*_{e^* l^*}\phi_{l^*}\right) - \sum_{e^* \in \mathcal V^*} E_e\left(\sum_{v^* \in \mathcal L^*}K_{v^*}B^*_{v^* e^*}\right).
  \end{equation}
  The transformation 
  \begin{equation}\label{eqn:lagdal}
    \begin{aligned}
      q_{v^*} &\rightarrow \phi_v \\ 
      \phi_{l^*} &\rightarrow q_l \\ 
      M^*_{v^* l^*} &\rightarrow - M_{l v} \\ 
      \mathcal C^* &\rightarrow \mathcal I \\ 
      \mathcal I^* &\rightarrow \mathcal C
    \end{aligned}
  \end{equation}
  can easily be seen to relate $L^*$ to $L$ after integrating the first term by parts.
\end{proof}

\begin{cor}\label{cor:lagtrans}
  Let $G$ be an embedded circuit with Lagrangian $L$ and $G^*$ be the dual circuit of $G$. By theorem \ref{thm:dualitylag} and \ref{thm:canquant}, it is always possible to find $N = |\mathcal V| - |\Gamma_I| - |\Gamma_C|-1$ variables so that the Lagrangian for $G$ may be written
   \begin{equation}
     L =  \sum_{i = 1}^{N} Q_i \dot \Phi_i - H(Q, \Phi).
  \end{equation}
  It is always possible to write the Lagrangian for $G^*$ as 
   \begin{equation}
     L^* = -\sum_{i=1}^N  \Phi_i \dot Q_i - H(-\Phi, Q).
  \end{equation}
\end{cor}
\begin{proof}
  Since 
  \begin{equation}
    \sum_{i = 1}^N Q_i \dot \Phi_i = \sum_{v,l} q_l M_{lv} \dot\phi_i,
  \end{equation}
  it follows that, under the transformation (\ref{eqn:lagdal}), 
  \begin{equation}
    Q_i \dot\Phi_i \rightarrow - \Phi_i \dot Q_i.
  \end{equation}
\end{proof}

We remark that, for nonplanar graphs, the Lagrangian transformations given in Cor. \ref{cor:lagtrans} is perfectly well--defined. It is also true that the transformation (\ref{eqn:lagdal}) is \emph{also} perfectly well--behaved. 
The issue is simply that, for nonplanar graphs, there is no systematic analogue of Defninition \ref{def:duality} that applies to nonplanar graphs in any meaningful sense. The reason for this is that edges in so--called topological loops would, in the dual of a nonplanar graph, can have an odd number of endpoints, which contradicts our definition of what an edge is. Nonetheless, it is sometimes possible to ``planarize" a nonplanar circuit and then take the dual after planarization, as we saw for $K_5$ in the main text.

\end{appendix}

\bibliography{thebib}
\end{document}